\documentclass[12pt]{iopart}

\usepackage{color}
\usepackage{cite}
\usepackage{amsthm}
\usepackage{amssymb}
\usepackage{braket}
\usepackage{wrapfig}
\usepackage{thmtools}
\usepackage{thm-restate}
\usepackage{hyperref}
\usepackage{cleveref}
\usepackage{graphicx}
\usepackage[font=large]{caption}
\usepackage{subcaption}
\usepackage{kbordermatrix}
\usepackage[T1]{fontenc}
\usepackage[utf8]{inputenc}

\renewcommand{\kbldelim}{[} 
\renewcommand{\kbrdelim}{]}

\newcommand{\Fi}{\mathcal{F}^{-1}}
\newcommand{\F}{\mathcal{F}}
\newcommand{\xx} {\mathbf{X}}
\newcommand{\yy} {\mathbf{Y}}
\newcommand{\zz} {\mathbf{Z}}
\begin{document}
\title[Time Independent Universal Computing with Spin Chains]{Time Independent Universal Computing with Spin Chains: Quantum Plinko Machine}

\author{K F Thompson$^1$, C Gokler$^2$, S Lloyd$^3$, 
and P W Shor$^4$}

\address{$^1$ School of Engineering and Applied Sciences, Harvard University, Cambridge MA 02138, USA}
\address{$^2$ School of Engineering and Applied Sciences, Harvard University, Cambridge MA 02138, USA}
\address{$^3$ Department of Mechanical Engineering, Room 3-160
     Massachusetts Institute of Technology
     77 Massachusetts Ave.
     Cambridge, MA 02139, USA}
\address{$^4$ Department of Mathematics, Room 2-375
     Massachusetts Institute of Technology
     77 Massachusetts Ave.
     Cambridge, MA 02139, USA}
\ead{kft@seas.harvard.edu}

\begin{abstract}
We present a scheme for universal quantum computing using XY Heisenberg spin chains.  Information is encoded into packets propagating down these chains, and they interact with each other to perform universal quantum computation.  A circuit using g gate blocks on m qubits can be encoded into chains of length $O(g^{3+\delta} m^{3+\delta})$ for all $\delta >0$ with vanishingly small error.
\end{abstract}
\pacs{03.65.Aa, 03.67.Ac, 03.67.Hk}

\maketitle

\section{Introduction}
Heisenberg spin chains have many applications in quantum computing.  Primarily they are used as quantum ``wires'' that can communicate a state from one part of a quantum computer to another \cite{Bayat2010, Benjamin2001, Burgarth2005_2, Burgarth2005, Fitzsimons2006, Gong2007, Osborne2004, Wojcik2005, Matthias2004}.  This is an attractive approach to quantum communication becuase it could avoid the problem of converting ``stationary'' qubits into ``flying'' qubits.  Certain types of quantum systems are easy to control and manipulate (NMR), while other types of quantum systems are more difficult to control (photons), but are less susceptible to error during transmission.  Mechanisms are designed to convert ``stationary'' qubits and ``flying'' qubits for purposes of communication between different parts of a quantum computer \cite{Kosaka2008}.  The benefit of having a spin chain as a quantum wire is that it allows one to directly communicate the information, without converting it to another form.  Of particular interest are schemes where information can be transmitted passively, with no control on the overall system, or control only over a small portion of the system\cite{Burgarth2005, Burgarth2005_2, Burgarth2007, Matthias2004, Osborne2004, Wojcik2005}.  For a comprehensive review of this area, please see \cite{Sougato2007}.  

We are particularly interested in the results of \cite{Osborne2004}.  In this paper, the author shows that we can initialize the quantum state of a spin chain to form a Gaussian packet.  If the ``momentum'' of these packets are chosen carefully, they will propagate around the chain in a near dispersion free fashion.  Then, up to some approximation, the spin chain acts like a wire for quantum information.  We will leverage this result in our scheme, as well as provide a new proof that these packets are dispersion free.

There are a number of researchers interested in using these spin chains to implement universal quantum computation \cite{Burgarth2009, Burgarth2010, Fitzsimons2006, Kempe2002, Childs2013, Janzing2007}.  In many of these papers, time dependent control on the system is needed to implement computations.  Taking a few examples, Burgarth has shown that any unitary on a spin chain can be implemented using controls on only the first two qubits\cite{Burgarth2010}, while Benjamin has given a scheme for universal computing using some ``always on'' interactions, and a time independent ``switch''\cite{Benjamin2001}.  Of particular interest to us is the result \cite{Childs2013}.  In this paper, Childs et al. have given a time independent Hamiltonian (for a multiparticle quantum walk) that implements any quantum computation with asymptotically vanishing error.

We provide a construction similar to that of Childs et al. using Heisenberg spin chains.  This approach provides an encoding of any standard qubit space $\left[\mathbb{C}^2 \right]^{\otimes n}$ into a superposition of wavepackets propagating on periodic spin chains of size N.  A time independent Hamiltonian can then be used to implement any circuit (up to some small error dependent on N) on the encoded information.  In contrast with \cite{Childs2013}, we use very large but very weak gates.  This is advantageous because it allows a simplified perturbative analysis as well as better error scaling.  If we let $g$ be the number of gate blocks and $m$ be the number of qubits, we achieve vanishing error for $N=O(m^{3+\delta} g^{3+\delta})$ for all $\delta > 0$.  Our result is not directly comparable with \cite{Childs2013}, since our controlled phase gate cannot be written as a multiparticle walk interaction, but we believe that this result shows a perturbative approach may be more useful for the problem of universal computation with a multiparticle quantum walk.
\section{Background}
In this work, we have set coupling constants and $\hbar$ to 1.  All of our results hold with these constants included, they have been removed for notational simplicity.  Also note that the number of encoded qubits will always be denoted m, and the number of gate blocks will always be denoted g.  

Consider a periodic chain of distinguishable spin 1/2 particles.  Numbering the sites from 0 to N-1, the Hilbert space is the span of all possible spin configurations of the system.  All such vectors have the form $\ket{a_0, ..., a_{N-1}}$ where $a_i \in \{ \downarrow, \, \uparrow\}$.  The first arrow in the string represents the state of spin $0$, the second represents the state of spin 1, the third represents the state of spin 2, etc.  Define the jth excitation subspace to be the subspace spanned by vectors with exactly j $\uparrow$ while the rest are $\downarrow$.  The 0th excitation subspace is 1 dimensional $\ket{\downarrow \downarrow ... \downarrow}$, and we will call this the vacuum state.  For the single excitation subspace, define $\ket{x}$ to be the vector where all spins are down except at position $x$ which is up: $\ket{\downarrow \downarrow ... \downarrow \uparrow \downarrow ... \downarrow}$.

On this ring, we are interested in the nearest neighbor XY Hamiltonian.  Letting the Pauli matrices be defined as they usually are:
\begin{equation}
\xx=\frac{1}{2}
\kbordermatrix{
    \mbox{} & \uparrow & \downarrow\\
    \uparrow & 0 & 1 \\ 
   \downarrow & 1 & 0
}\,\,
\yy=\frac{1}{2}
\kbordermatrix{
    \mbox{} & \uparrow & \downarrow\\
    \uparrow & 0 & -i \\ 
   \downarrow & i & 0
}\,\,
\zz=\frac{1}{2}
\kbordermatrix{
    \mbox{} & \uparrow & \downarrow\\
    \uparrow & 1 & 0 \\ 
   \downarrow & 0 & -1
}
\end{equation}

we define the nearest neighbor Hamiltonian to be:
\begin{equation}
H=H_{{ \rm ring}}^1=2\sum_{j=0}^{N-1} \mathbf{X}_{j} \mathbf{X}_{j+1} +\mathbf{Y}_j\mathbf{Y}_{j+1}=\sum_j \kbordermatrix{
\mbox{} & \uparrow \uparrow & \uparrow \downarrow & \downarrow \uparrow & \downarrow \downarrow \\
\uparrow \uparrow & 0 & 0 & 0 & 0\\
\uparrow\downarrow & 0 & 0 & 1 & 0\\
\downarrow \uparrow & 0 & 1 & 0 & 0\\
\downarrow \downarrow & 0 & 0 & 0 & 0
}
\end{equation}
where $\mathbf{X}_j$ is acting on the jth site, meaning that $\xx_j=\mathbb{I} \otimes ... \otimes \mathbb{I} \otimes \xx \otimes \mathbb{I} \otimes... \otimes \mathbb{I}$, and where the $Nth$ site is understood to be the same as site 0 (periodicity).

The first thing to notice is that the vacuum state $\ket{{\rm vac}}$ has energy $0$.  It is not actually the ground state of the system, some states have negative energy.  The second thing to notice is that the total number of excitations is conserved.  More formally, the operator $\mathbf{Z}_{{ \rm tot}}=\sum_{j=0}^{N-1} \mathbf{Z}_j$ commutes with the Hamiltonian, so the Hamiltonian can be diagonalized with eigenvectors of $\mathbf{Z}_{{ \rm tot}}$.  In particular, the Hamiltonian can be diagonalized in subspaces with some constant number of excitations (eigenspaces of $\mathbf{Z}_{{ \rm tot}}$).  The full set of eigenstates has been well studied \cite{Bethe1931}, however we will be interested in the single excitation eigenvectors.  For $p \in \{0, 1, ..., N-1 \}$ they are of the form:
\begin{equation}\label{momentum}
\ket{p}_{m}=\frac{1}{\sqrt{N}}\sum_{x=0}^{N-1} e^{\frac{2 \pi i p x}{N}}\ket{x}
\end{equation}
with eigenvalue $E(p)=2\cos\left(\frac{2 \pi p}{N}\right)$.  Since these states are orthonormal, and since there are $N$ of them, they constitute a basis for the single excitation subspace.  Consequently, we can expand any vector in this space as a linear combination of this basis or a linear combination of the standard computational basis:
\begin{equation}
\ket{\psi}=\sum_{x=0}^{N-1} A_x \ket{x}=\sum_{p=0}^{N-1} a_p \ket{p}_m
\end{equation}
It is easy to see that the coefficients are related via the the finite Fourier transform:
\begin{equation}
\fl a_p=\F[A](p)=\frac{1}{\sqrt{N}}\sum_{x=0}^{N-1} A_x e^{-\frac{2 \pi i p x}{N}} { \rm   and  } A_x= \Fi[a](x)=\frac{1}{\sqrt{N}}\sum_{p=0}^{N-1} a_p e^{\frac{2 \pi i p x}{N}}
\end{equation}

We can interpret these states as momentum (scattering) states in the discrete setting, and take the derivative of the energy function to obtain the group velocity for a packet narrowly distributed around some value $p_0$.  We think about these chains in the continuous limit, so $\frac{2 \pi p}{N}\approx k$.   In this case:
\begin{equation}
v_g(p)=\frac{dE}{dk}=-2\sin \left( \frac{2 \pi p}{N}\right)
\end{equation}
There are two things to notice about this expression.  The first is that packets travel at constant velocity, and the second is that near $p_0=N/4$ the group velocity is constant up to second order, so we expect these packets to propagate around the chain in a near dispersion free fashion\cite{Osborne2004}.

For our waveforms, we will use a finite (N-periodic) Gaussian\cite{Cotfas2012}.  On our chain this is defined as:
\begin{equation}
G_{\Delta x}=\frac{1}{\sqrt{\Delta x\sqrt{\pi}}}\sum_{x=0}^{N-1} \sum_{\alpha=-\infty}^\infty e^{-\frac{(\alpha N+x-x_0)}{2\Delta x^2}} \ket{x}
\end{equation}
Let $\Delta x$ and $\Delta p$ be positive real numbers that satisfy $2 \pi \Delta x \Delta p=N$.  The packets we have defined (approximately) satisfy a discrete Heisenberg relation:
\begin{eqnarray}
 \fl \frac{1}{\sqrt{\Delta x\sqrt{\pi}}}\sum_{x=0}^{N-1} \sum_{\alpha=-\infty}^\infty e^{\frac{2 \pi i p_0 x}{N}}e^{-\frac{(\alpha N+x-x_0)^2}{2\Delta x^2}} \ket{x}= \nonumber \\
 \frac{1}{\sqrt{\Delta p \sqrt{\pi}}} \sum_{p=0}^{N-1} \sum_{\alpha=-\infty}^\infty e^{-\frac{2 \pi i p x_0}{N}} e^{-\frac{(\alpha N+p-p_0)^2}{2\Delta p^2}}\ket{p}_m
\end{eqnarray}

This state represents a Gaussian packet centered at position $x_0 \,\,{\rm mod}\,\, N$ on the chain and with momentum centered at $p_0 \,\,{\rm mod}\,\, N$.  We expect this packet to approximately translate on the ring with speed $v_g(p_0)$.  After some time $t$, we expect the packets position to be approximately $x_0-v_g(p_0)t$.  For technical reasons, we will assume throughout the paper that $p_0=N/4$ is an integer, while $x_0 \in \,[0, N-1]$ may not be an integer.

We will also make use of some well known results on universality.  Suppose we have a system of n qubits.  It is known that if we can perform an arbitrary single qubit gate on any of the qubits and a controlled phase between any two qubits, then we can implement any unitary on the system \cite{NielsenChuang2011}.  It is also known \cite{NielsenChuang2011} that (up to an overall phase) we can write any single qubit unitary on a system in the form:
\begin{equation}
U=e^{i \theta_3 \zz} e^{i \theta_2 \xx } e^{i \theta_1 \zz}
\end{equation}
We will use these two facts to show our scheme is universal.  

\section{Description of Scheme}
To illustrate our basic idea, consider a Z gate on a single encoded qubit, simulating a quantum circuit on one qubit.  We will encode our circuit using the ``dual rail'' encoding.  In this encoding each logical qubit corresponds to two rails, a 0 rail and a 1 rail.  We will encode the $\ket{0}$ state into a packet propagating down chain 0 and the vacuum state on the other chain, and encode the $\ket{1}$ state into a packet propagating down chain 1 with a vacuum on the other chain.  Any superposition of these states is then realized as a superposition of packets propagating along the rails.

Now imagine the Hamiltonian on the chain 1 has an additional term with adds a small extra phase per unit time.  Let the Hamiltonian for this chain be:
\begin{equation}
\sum_{j=0}^{N-1} \xx_j \xx_{j+1} +\yy_j \yy_{j+1} + \Phi \left(\zz_j+\frac{1}{2}\mathbb{I}_j \right)
\end{equation}
where $\mathbb{I}$ is the standard identity matrix, and $\Phi << 1$.  Now an encoded $\ket{0}$ state will be unaffected by the extra term, since the extra term acts as $0$ on the vacuum state.  On chain 1, however, the momentum states are still diagonal with eigenvalues $2\cos\left(\frac{2 \pi p}{N} \right)+\Phi$. The encoded $\ket{1}$ state will gain an extra phase of $\Phi$ per unit time.  After some time $t$, the $\ket{1}$ state has gained an extra phase $e^{i\Phi t}$.  Effectively, we are implementing an encoded $e^{i\theta \zz}$ operation on the encoded qubit.

\begin{wrapfigure}{r}{0.5\textwidth}
  \begin{center}
    \includegraphics[width=0.48\textwidth]{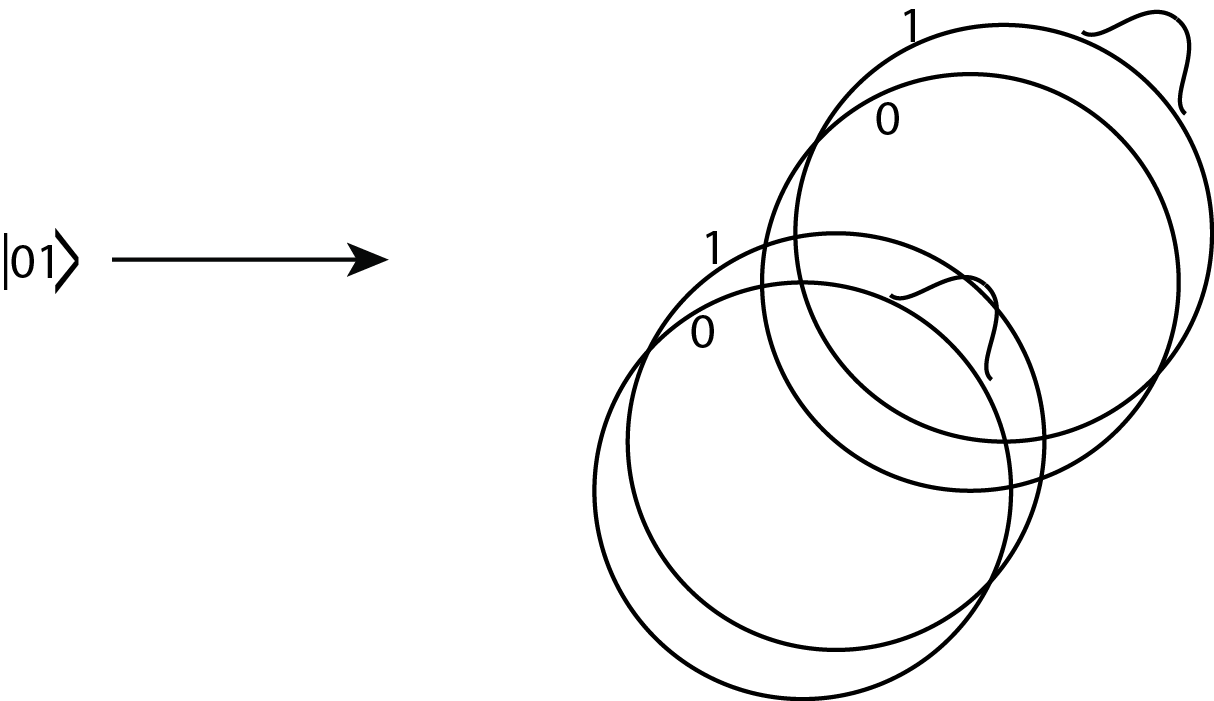}
  \end{center}
  \caption{Dual rail encoding of a computational basis state}\label{dualrail1}
\end{wrapfigure}

Suppose now that the chain is very long and that only a portion of the chain has the added term causing the encoded $\zz$ gate.  We expect that when the packets are well localized outside the gate, they will not ``see'' the added potential.  So, up to some approximation, they will propagate along the chain without gaining extra phase until they reach the edge of the virtual gate.  Now, if the gate is much larger than the packet and the packet is well localized inside the gate, we expect it will appear to the packet as though the gate spans the whole chain (again up to some approximation).  The remaining issue is the ``transient'' regime where the packet is transitioning from outside the gate to being inside the gate.  Fortunately, provided the gate strength is weak ($\Phi <<1$) we can employ a bound from \cite{Loan1977} to show that the transient regime has a very small effect on the propagating packets.

Note that for this gate, and for the gates that follow, the gate Hamiltonian only commutes with the ring Hamiltonian when the gate spans the entire ring.  When the gate is located in someregion of the qubit ring, $[H_{\rm ring}, H_{\rm gate}]$ is some small operator supported at the edges of the gate.

\subsection{Encoding}

Now we generalize some of the previous discussion to the simulation of many qubits.  Much like in the previous section, a dual rail encoding \cite{NielsenChuang2011} \cite{Childs2013} is employed.  Consider a computation with g gates on m qubits.  Given 2 rails for each qubit let one correspond to the ``0'' rail and the other to the ``1'' rail.  Let us index the rails with two numbers.  The first number is a 0 or 1 depending on whether we are talking about the 0 or the 1 rail.  The second number indicates identifies the specific qubit.  The 0 rail for the jth qubit would be denoted $(0, j)$.  For Pauli operators, we will use the superscript to denote which rail the Pauli operator acts on, and a subscript to denote which site on the rail the Pauli operator acts on.  A Pauli operator acting on the 0th rail of the 3rd qubit at the 5th site will be denoted $\mathbf{X}_{5}^{(0, 3)}$.  The nearest neighbor Hamiltonian for all $2m$ chains can then be written as:
\begin{equation}
H=H_{{ \rm rings}}^{2m} = \sum_{a=0}^1 \sum_{b=1}^m \sum_{c=0}^{N-1} \mathbf{X}^{(a, b)}_c \mathbf{X}^{(a, b)}_{c+1} + \mathbf{Y}^{(a, b)}_c \mathbf{Y}^{(a, b)}_{c+1}
\end{equation}

Now, let us encode computational basis states into Gaussian packet on the appropriate rails centered around $p_0=N/4$.  The state $\ket{01}$ would be encoded as (see \cref{dualrail1}):
\begin{eqnarray}
\ket{\psi}=\left(\sum_{p_1=0}^{N-1}\sum_{\alpha=-\infty}^\infty e^{-\frac{2 \pi i p_1 x_0}{N}}e^{-\frac{(\alpha N +p_1-p_0)^2}{2\Delta p^2}}  \ket{p_1}_m\right) \otimes \ket{{ \rm vac}} \otimes \ket{{ \rm vac}} \otimes  \nonumber \\
 \left( \sum_{p_2=0}^{N-1} \sum_{\alpha=-\infty}^\infty e^{-\frac{2 \pi i p_2 x_0}{N}}e^{-\frac{(\alpha N +p_2-p_0)^2}{2\Delta p^2}}\ket{p_2}_m\right)
\end{eqnarray}
A superposition of $\{\ket{00}, \ket{01}, \ket{10} \ket{11}\}$ would be encoded as a superposition of the packets previously described.  Packets will propagate down the chains and interact weakly with sets of gates and with each other.  We will describe methods for performing encoded CPHASE operations, as well as $e^{i \theta \zz}$ and $e^{i \theta \xx}$ in the following sections, thereby giving a scheme for universal computation.  

\subsection{Z Gate}

The $\zz$ gate has already been described, we will recap the discussion here briefly.  Our $\zz$ gate consists of a very long but very weak interaction placed on the 1 rail.  For each point inside the gate we add an extra term $\mathbf{Z} +\frac{1}{2}\mathbb{I}$.  Locally, the momentum eigenstates are still diagonal, so this has the effect of adding an additional phase of $\Phi$ per unit time.  If the gate is extended to the whole chain, we would be able to write the energy as $2\cos \left(\frac{2 \pi p}{N} \right)+\Phi$.

The packets travel at constant velocity (the group velocity of the packet), so the time a packet spends inside a gate is roughly proportional to the size of the gate.  If we ignore the transient regime, and we assume that a localized packet does not ``see'' the outside of the gate, then the overall phase gained from a gate is $\approx { \rm gate\, length} * \Phi$.  For computational universality, we need to show how to add any phase to our packet.  So, we need the gate strength $\Phi$ to be at least $\Omega \left(\frac{1}{{ \rm gate\, length}}\right)$.

\subsection{X Gate}
Again consider one encoded qubit and 2 rails.  Connect each point $j$ on the 0 rail to its adjacent point on the 1 rail with a coupling of the form $2\Phi( \xx^0_j \xx^1_j +\yy^0_j \yy^1_j)$, where the qubit has not explicitly been designated since there is only one.  The full Hamiltonian will then have the form:
\begin{equation}
H=H_{{ \rm ring}}^2+\sum_{j=0}^{N-1} 2 \Phi ( \xx^0_j \xx^1_j +\yy^0_j \yy^1_j)
\end{equation}
The extra interaction term commutes with the ring Hamiltonian, so there is a basis which diagonalizes both.  The relevant (single excitation) states are $\ket{\pm, p}=\ket{p}_m \otimes \ket{{ \rm vac}}\pm \ket{{ \rm vac}}\otimes \ket{p}_m$ with energy $E(\ket{\pm, p})=2 \cos \left(\frac{2 \pi p}{N} \right) \pm \Phi$.  We can construct encoded $\ket{\pm}$ states by taking appropriate linear combinations of the aforementioned states.  Assuming that the packets are centered at 0 these have the form:
\begin{equation}
\ket{\pm}=\sum_{\alpha, p} e^{-\frac{(\alpha N+p-p_0)^2}{2\Delta p^2}} \ket{\pm, p}
\end{equation}

An encoded $\ket{+}$ state gains an extra phase of $\Phi$ per unit time, and an encoded $\ket{-}$ state gains an extra phase of $-\Phi$ per unit time.  This type of gate creates a relative phase between encoded $\ket{\pm}$ states and thus implements an encoded $e^{i \theta \xx}$ gate.  The same scaling analysis applies.  We need $({ \rm gate \, strength}) \geq ({ \rm gate \, length})$ to implement some constant phase.

\subsection{Controlled Z Gate}
Now imagine we have two encoded qubits and we are interested in performing a controlled phase operation.  Connect each site on the $(1, 1)$ rail with every other site on the $(1, 2)$ rail with an interaction of the form 
\begin{equation}
H_{{ \rm int}}=\Phi(\zz +\frac{1}{2}\mathbb{I}) \otimes (\zz+\frac{1}{2}\mathbb{I})=
\kbordermatrix{\mbox{} & \uparrow \uparrow & \uparrow \downarrow & \downarrow \uparrow & \downarrow \downarrow\\
\uparrow \uparrow & 1 & 0 & 0 & 0\\
\uparrow \downarrow & 0 & 0 & 0 & 0\\
\downarrow \uparrow & 0 & 0 & 0 & 0\\
\downarrow \downarrow & 0 & 0 & 0 & 0
} 
\end{equation}
In other words, the new Hamiltonian is:
\begin{equation}
H=H_{{ \rm ring}}^4+ \Phi\sum_{i \leq j =0}^{N-1} (\zz_i^{(1, 1)} +\mathbb{I}_i^{(1, 1)}) \otimes (\zz_j^{(1, 2)} +\mathbb{I}_j^{(1, 2)})
\end{equation}

Just as in the phase gates, the momentum states are still diagonal.  Ordering the rails $(0, 1), (1, 1), (1, 2), (0, 2)$, momentum states of the form $\ket{{ \rm vac}} \otimes \ket{p_1} \otimes \ket{p_2} \otimes \ket{{ \rm vac}}$ have eigenvalues $2 \cos \left( \frac{2 \pi p_1}{N} \right)+2 \cos \left( \frac{2 \pi p_2}{N}\right) + \Phi $.  The other momentum states are unchanged.  In other words, the encoded $\ket{11}$ state gains an extra phase of $\Phi$ per unit time, and nothing happens to the other encoded states.  We can thus perform an encoded controlled phase operation on our qubits.  

\begin{figure}
\centering
\begin{subfigure}{.5\textwidth}
  \centering
  \includegraphics[width=0.8\linewidth]{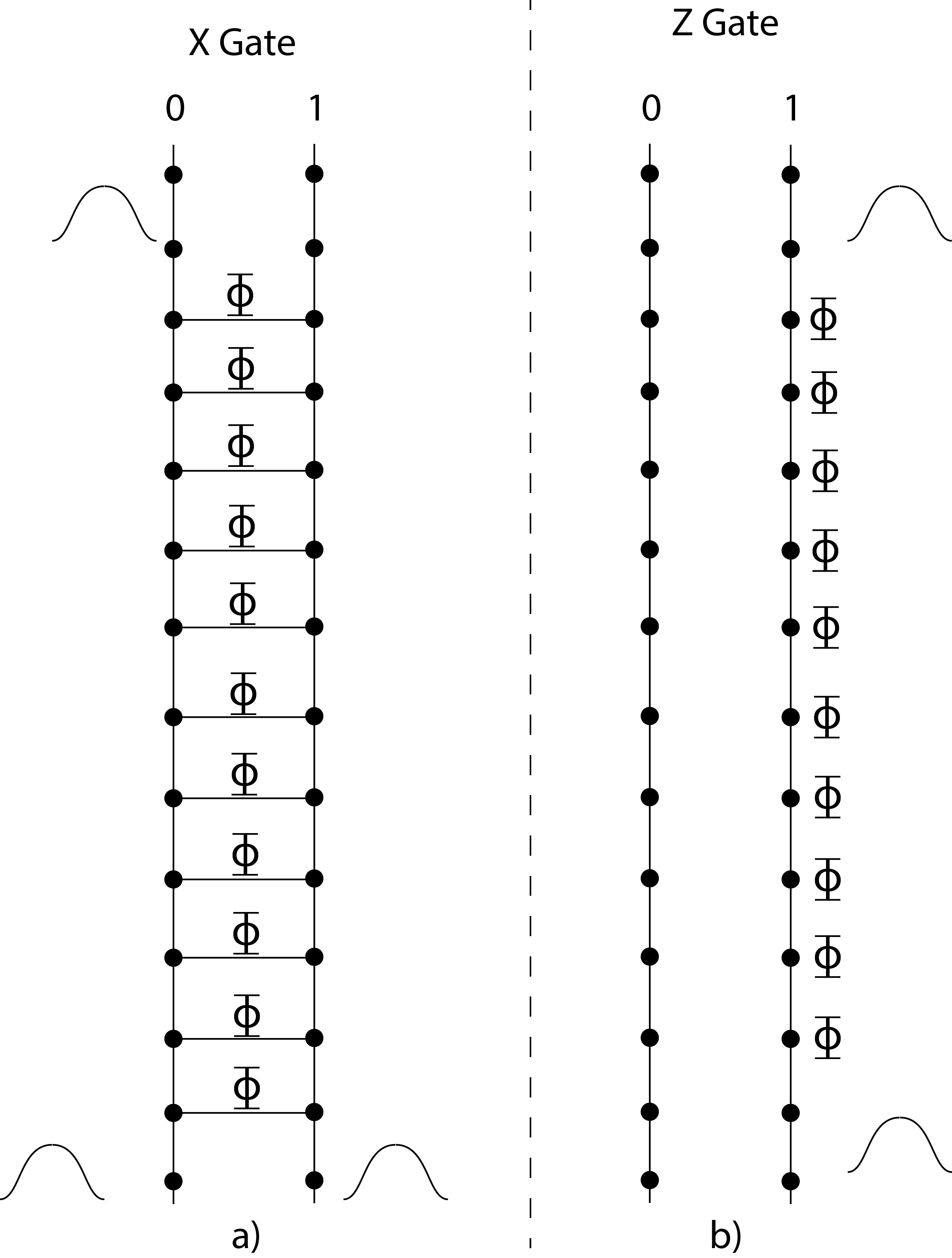}
  \caption{Our proposed $\xx$ and $\zz$ gates}
\end{subfigure}%
\begin{subfigure}{.5\textwidth}
  \centering
  \includegraphics[width=.9\linewidth]{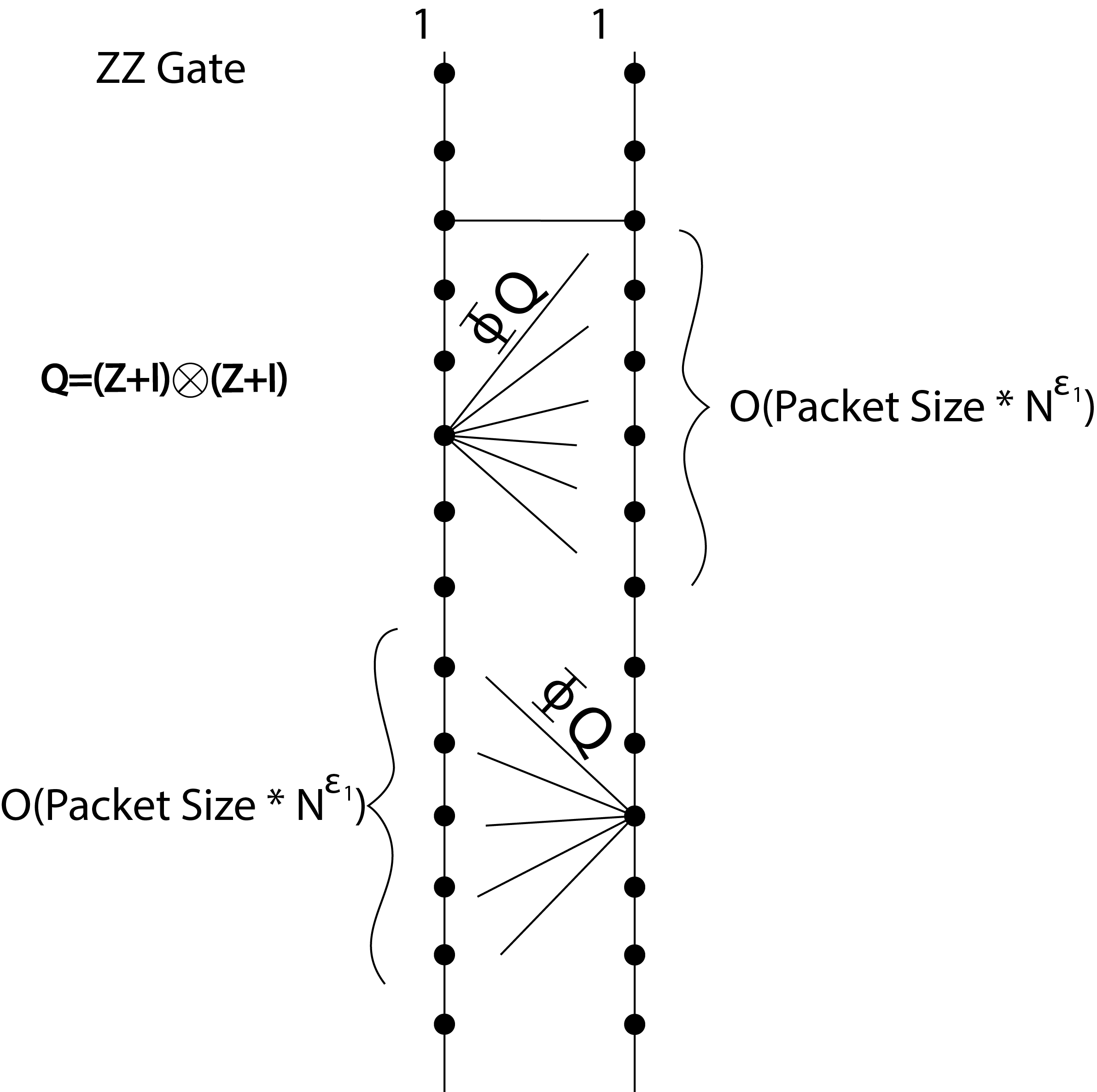}
  \caption{Our proposed CPHASE gate}
  \label{cphase}
\end{subfigure}
\caption{Proposed Gates}
\end{figure}

It may be unclear how exactly to implement a gate like this.  After all, needing to connect every point on one rail to every other point on the other rail is a very non local interaction.  It is hard to see exactly how to truncate the Hamiltonian so that a localized packet will only ``see'' this interaction.  For this we use the following (see \cref{cphase}).  Connect each point on the $(1, 1)$ rail with the point adjacent to it and to the points on the rail $(1, 2)$ that are a distance {\em (packet size)*(some fractional power of N away)}.  If the packets are well localized, they should not ``see'' any missing interactions.

\subsection{Gate Blocks}
Our gate blocks will be defined in the following way.  Given a description of some circuit using phase and single qubit gates on m qubits, partition the gates into operations that can be done simultaneously.  If the circuit is on four qubits and the first gates are an x gate on the first qubit, a controlled phase gate on qubits 2 and 3 and a z gate on qubit 4, we would group these operations together, since they can  be done simultaneously.  Our first gate block then would consist of these virtual operations on the rails (see \cref{block1}).  Immediately after this, place the second group of simultaneous operations, block 2 \cref{blocks}, and proceed until the entire circuit has been applied.

\begin{figure}
\centering
\begin{subfigure}{.5\textwidth}
  \centering
  \includegraphics[width=1\linewidth]{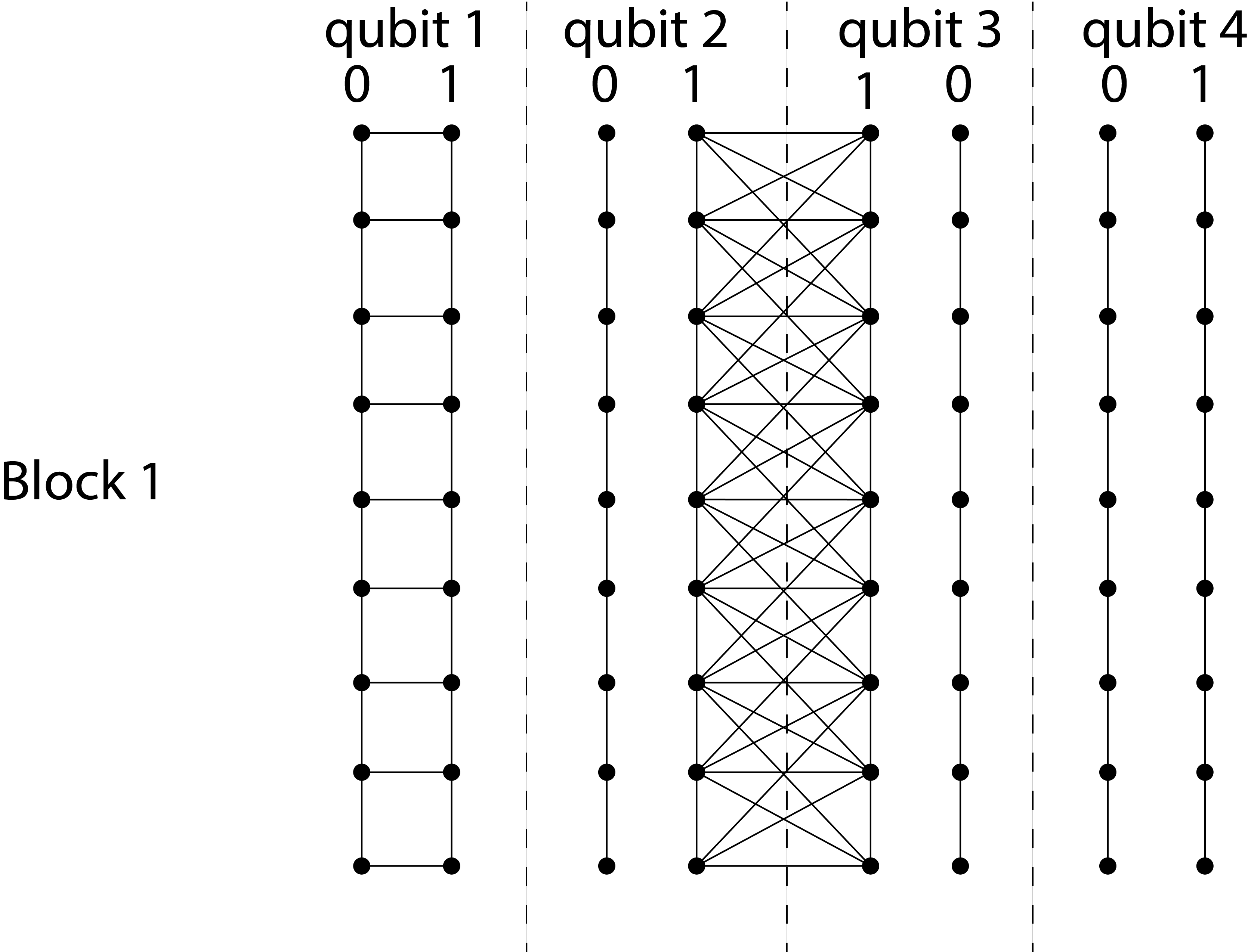}
  \caption{Block 1: the first group of operations}
  \label{block1}
\end{subfigure}%
\begin{subfigure}{.5\textwidth}
  \centering
  \includegraphics[width=.4\linewidth]{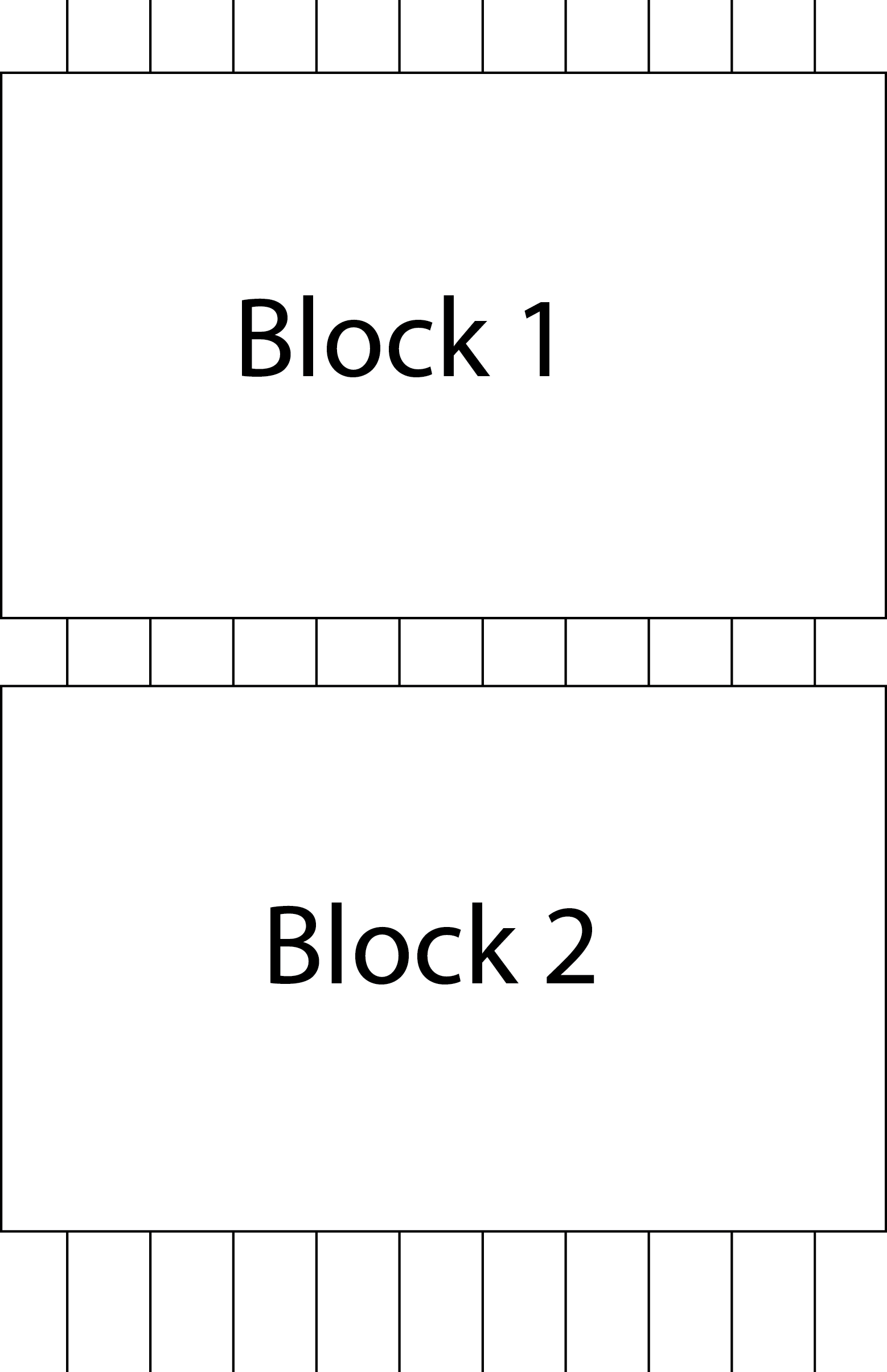}
  \caption{The composition of two blocks}
  \label{blocks}
\end{subfigure}
\caption{Gate Blocks}
\end{figure}

By varying $\Phi$ we can have the gates discussed implement encoded $e^{i \theta \xx}$, $e^{i \theta \zz}$, and $e^{i \theta (\zz+\frac{1}{2} \mathbb{I})\otimes (\zz+\frac{1}{2}\mathbb{I})}$ operations for any $\theta$.  We can therefore accomplish universal computing in this way.

\section{Description of Bounds}
In this section we will provide a statement of our bounds, as well as a brief description of our strategy for obtaining them.  Formal proofs can be found in the appendix.  

As was mentioned already, there are three aspects of our scheme to worry about.  The first is that our packets are designed to be dispersion free up to second order, but still do in fact have some small higher order dispersion.  In our final scaling, we will take $\Delta p \approx N^{2/3}$.  Since the full range of $p$ is $\{0, ..., N-1 \}$ in the limit of large $N$, the spread in momentum will be a vanishingly small fraction of the allowed values of $p$.  So, in the limit of large $N$ we expect the group velocity to get closer and closer to being constant, and thus the packets to get closer and closer to being dispersion free.  The bound we obtain is:

\begin{restatable}[Dispersion Result]{thm}{disp} \label{thm:disp}
Given a single ring with nearest neighbor XY couplings $\displaystyle H_{{ \rm ring}}^1=2\sum_{j=0}^{N-1}\mathbf{X}_j \mathbf{X}_{j+1} +\mathbf{Y}_j \mathbf{Y}_{j+1}$.  Let $U_{2t}$ translate the the packets a distance 2t (time multiplied by the group velocity of the packet).  Then if $\ket{\psi}$ is a wavepacket centered at  $x_0$ in x space and $p_0=N/4$ in p space, 
$$
\ket{\psi}=\frac{1}{\sqrt{\Delta p \sqrt{\pi}}}\sum_{p=0}^{N-1} \sum_{\alpha=-\infty}^\infty e^{-\frac{2 \pi i p x_0 }{N}} e^{-\frac{(\alpha N +p-p_0)^2}{2\Delta p^2}} \ket{p}_m
$$
we have:
$$
\left| \left| (e^{-i H_{{ \rm ring}} t}-U_{2t})\ket{\psi}\right| \right| =O\left(\frac{t \Delta p^3}{N^3}\right)
$$
\end{restatable}

We can see from this result that if the packet has a spread $N^{2/3-\varepsilon}$ in momentum space ($N^{1/3+\varepsilon}$ in x space), the packet will make it all the way around the chain with vanishingly small dispersion in the limit of large N.  If $\Delta p=N^{2/3-\varepsilon}$ and the packet completes a revolution of the chain ($2t=N$), then the packet is dispersion free up to an error $O\left( \frac{(N^{2/3-\varepsilon})^3}{N^3}\right)=O \left( \frac{1}{N^{3\varepsilon}}\right)$.  We will need the following theorem proven in \cite{Kitaev2002} and \cite{Childs2013} to extend this analysis to multiple chains:
\begin{restatable}[Linear Propagation of Error]{thm}{kitaev}\label{thm:kitaev}
Let $U_1, ..., U_n$ and $V_1, ..., V_n$ be unitary operators.  Then for any $\ket{\psi}$,
\begin{equation}
\left| \left| \left(\prod_{i=n}^1 U_i-\prod_{i=n}^1 V_i \right)\ket{\psi}\right| \right| \leq \sum_{j=1}^n \left| \left| (U_j-V_j)\prod_{i=j-1}^1 U_i \ket{\psi}\right| \right|
\end{equation}
\end{restatable}
Using the two above theorems, we can find the effect of dispersion on m rings.  
\begin{restatable}{cor2}{mult}\label{disp_multiring}
Suppose that we have $m$ rings with Hamiltonian $H_{{ \rm ring}}^m$ and a tensor product $\ket{\psi}$ of $m$ packets all centered at some position $x_0$ on the rings.  Suppose $p_0=N/4$ for each packet and let $U_{2t}$ be the unitary that translates all the packets a distance $2t$.
Then we have that, 
\begin{equation}
\left|\left| \left(e^{-iH_{{ \rm ring}}^m t}-U_{2t}\right)\ket{\psi}\right| \right| =O\left(\frac{mt \Delta p^3}{N^3} \right)
\end{equation}

\end{restatable}
\noindent This easily extends to a superposition in the dual rail encoding.  In this case exactly the same scaling holds with $H_{{ \rm ring}}^m$ replaced with $H_{{ \rm ring}}^{2m}$

The next important bound concerns the ``transient'' regime, where a superposition of packets is entering or leaving a gate.  In this regime, a packet is not really localized within any particular gate block.  Formally, we treat this situation by showing that we can ignore these interactions.  If the gate strength is weak enough, we can translate the packets from outside to inside the gate without picking up much error.  Note that by skipping over a portion of the gate, we are losing out on some of the phase implemented by that particular gate block.  However, the error associated with the lost phase is the same order as the error accumulated from skipping over the transient region.  As a result, it will be negligible in the final scaling (see appendix).  For the analysis, we employ a sensitivity result \cite{Loan1977}.  
\begin{restatable}[Matrix Exponential]{thm}{matexp}\label{loan_thing}
Let $A$ and $E$ be skew Hermitian complex matrices, and let t be a real scalar.  It holds that:
\begin{equation}
\left| \left|e^{(A+E)t}-e^{At} \right| \right| \leq ||E||t e^{||E||t}
\end{equation}
where the norm used is the standard operator norm (largest eigenvalue)
\end{restatable}
\noindent The result we obtain is:
\begin{restatable}[Bound for Transient Regime]{cor2}{transregime} \label{transregime}
Let $H_{{ \rm ring}}^{2m}$ be the ring Hamiltonian, and let $\hat{H}$ be the gate Hamiltonian.  Then if $\ket{\psi}$ is a superposition of packets in the dual rail encoding, it holds that:
\begin{equation}
 \left| \left|\left( e^{-i(H_{{ \rm ring}}^{2m}+\hat{H})t}-U_{2t}\right)\ket{\psi} \right| \right| =O\left(mt\Phi+\frac{m t \Delta p^3}{N^3}\right)
\end{equation}
\end{restatable} 
\begin{proof}
see appendix.  
\end{proof}

The main result of this paper is an error bound on the gates themselves.  We show that, up to some approximation, the gates translate the packets, as well as add a phase.  We show this by approximating the circuit Hamiltonian with a Trotterized Hamiltonian \cite{Lloyd1996}, and bound each of the elements in the Trotterization that correspond to gates that are ``far away'' from the packets.  With only the remaining (current) gates left, we can again approximate with a Trotterization and show that the Hamiltonian effectively translates the packets and adds a phase.  The main result is:

\begin{restatable}{thm}{Trotter} \label{steve}
Suppose we have a superposition of packets localized inside some gate block.  Suppose the packets remain localized to a distance $d$ for all $2 t'$ such that $0 \leq 2 t' \leq 2 t$.  Suppose in addition to the current gates, there are $k=O(poly(N))$ additional gates corresponding to gates before or after the current ones with Hamiltonians $H_1$, $H_2$, ..., $H_k$.  Then for sufficiently high degree polynomials $n(N)=N^{q_1}$ and $n'(N)=N^{q_2}$:
\begin{eqnarray}
\left|\left|\left( e^{-i(H_{{ \rm current}}+H_1+...+H_k)t}-U_{2t}^p\right)\ket{\psi}\right|\right|=O\bigg{(}poly(N)e^{-\frac{d^2}{2\Delta x^2}}+\\
\frac{m t \Delta p^3}{N^3}+\frac{m^2 t^2}{nn'} + \frac{t^2 m^2}{n'}\bigg{)}
\end{eqnarray}
\end{restatable}
\begin{proof}
See appendix
\end{proof}
Note that the dominant term is the $\frac{m t \Delta p^3}{N^3}$ dispersion.  Essentially we show that for localized packets, we can completely ignore the terms that are far away from the packet center.

\section{Circuit Error Bound}
The total error scaling from the whole circuit must be analyzed.  In each step, \cref{thm:kitaev} will be implicitly applied.  Suppose there are two adjacent gate blocks, and the first one applies a unitary (phases on some basis) on the encoded information and a translation across the gate up to an error $e_1$, while the other one applies a unitary and a translation up to $e_2$.  Then, the combination of the gate blocks applies a the net unitary and the full translation, up to an error $e_1+e_2$ according to \cref{thm:kitaev}.

\begin{figure}[h!]
  \centering
    \includegraphics[width=0.6\textwidth]{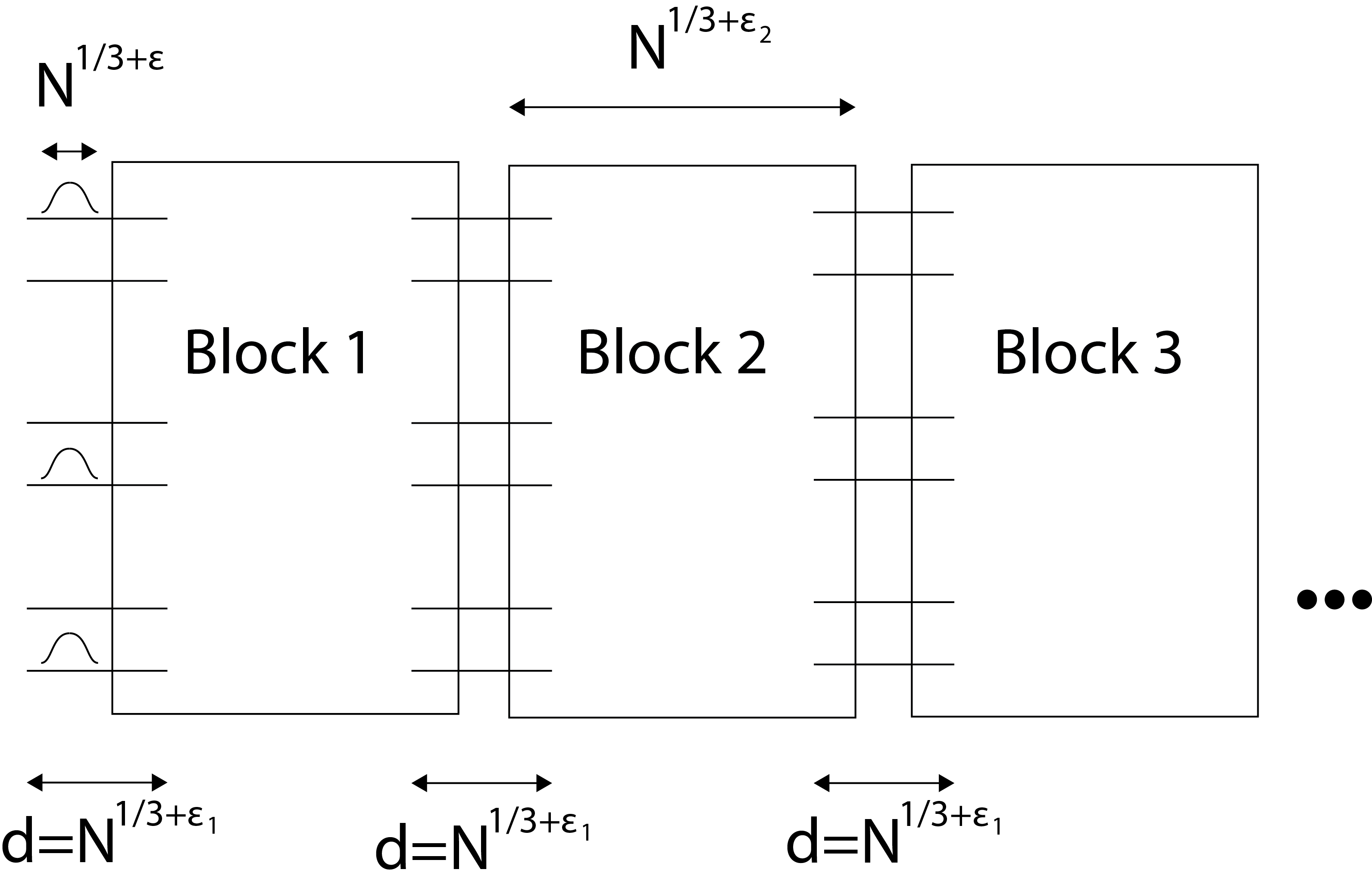}
    \caption{Relevant Length Scales}
\end{figure}

The basic strategy is as follows.  We will initialize our packets outside the first gate block, some distance $\Theta (d)$ away from the entrance.  We will then use the transient bound to move the packets a distance $\Theta(d)$ inside the gate block.  Next, we will use our Trotter bound to show that the packets approximately propagate through the gate region and pick up the appropriate phase.  From here, we will use the transient bound to translate the packets into the next block and repeat.

There are several important parameters that need to be specified in our analysis.  Let the length of the packets be $\Theta(N^{1/3+\varepsilon})$, the truncation length $d=\Theta(N^{1/3+\varepsilon_1})$, the Trotterization $n$, and the length of the gate $\Theta(N^{1/3+\varepsilon_2 })$.  In order for our analysis to make sense, it must be true that $\varepsilon< \varepsilon_1 < \varepsilon_2$, since the packet must fit inside the truncated region.  It will be argued that sections of the size of the truncation region can be ignored, so it must be that $\varepsilon_1 < \varepsilon_2$ otherwise we would be arguing that our gates can be ignored.  In the scaling analysis, let $\Phi=\Theta \left(\frac{1}{N^{1/3+\varepsilon_2}}\right)$, for reasons already discussed.

We will apply the transient bound over an interval $O(N^{1/3+\varepsilon_1})$ at most $O(g)$ times.  The Trotter bound is needed once per gate.  The time interval (assuming the gate is $O(N^{1/3+\varepsilon_2})$ long) is $O(N^{1/3+\varepsilon_2})$, and we will need it $O(g)$ times as well.  Assuming n and $n'$ are large enough polynomials is $N$, we can ignore terms with denominator $n$ and $n'$.
\begin{equation}
\fl O\left(poly(N) e^{-\frac{d^2}{2\Delta x^2}}+\frac{N^{1/3+\varepsilon_2}mg\Delta p^3}{N^3}+m g N^{1/3+\varepsilon_1}\Phi+\frac{m g N^{1/3+\varepsilon_1}\Delta p^3}{N^3}\right)
\end{equation}

Since we are assuming that $\varepsilon_1 > \varepsilon$, the first term will be exponentially small, independent of how big the polynomial in front of it is.  Once we substitute $\Phi=\Theta\left( \frac{1}{N^{1/3+\varepsilon_2}} \right)$, the remaining terms are:
\begin{equation}
O\left(\frac{mg}{N^{2/3+3\varepsilon-\varepsilon_2}}+\frac{mg}{N^{\varepsilon_2-\varepsilon_1}}\right)
\end{equation}
The best scaling with respect to both $m$ and $g$ is easily seen to be is obtained when $\varepsilon=0$, $\varepsilon_1 > 0$ and $\varepsilon_2=1/3$.  In this case, we can obtain error $O\left(\frac{1}{(mg)^{\delta/3-\varepsilon_1}} \right)$ with $N=\Omega(m^{3+\delta} g^{3+\delta})$ for any $\delta > 0$.

\section{Conclusion}

There are a number of features to notice in our scheme.  First, note from our bounds that we are effectively only limited by the dispersion of the packets.  It is conceivable that if we choose waveforms with even less dispersion that the result could be improved.  Second, there is a close analogy between our scheme and standard optical quantum computing approaches\cite{Ralph2011}.  Our CPHASE gate is effectively a nonlinear material that adds phase onto the overall quantum state if two photons are present.  We even use dual rail encoding, which is standard in many of these optical schemes.

We have provided a scheme for universal quantum computation using spin chains.  Given any quantum circuit written as single qubit and CPHASE gates, we have provided a time independent Hamiltonian that simulates the circuit, with asymptotically negligible error.  In contrast to other known schemes \cite{Childs2013}, we use weak gates that slowly effect our propagating packets over a long period of time.  This allows a simple perturbative analysis to succeed in analyzing the error scaling. 

It would be interesting to analyze the scattering of our Gaussian packets.  If we could scatter the packets off one another to implement a controlled phase, we would be able to use our scheme as a universal quantum computer without resorting to our complicated CPHASE gates.  If our analysis carried through, we could explicitly demonstrate better scaling for universal computation using multiparticle quantum walks.

Of independent interest is the approximate Heisenberg relation, and the approximately dispersion free result.  While these concepts have been noted by other authors \cite{Osborne2004, Cotfas2012}, we have provided a powerful new application for them.  These may allow for the construction of other interesting schemes using propagating Gaussian packets on qubit rings.  

\section*{References}

\nocite{*}
\bibliography{bibliography}{}

\providecommand{\newblock}{}
\begin{thebibliography}{10}
\expandafter\ifx\csname url\endcsname\relax
  \def\url#1{{\tt #1}}\fi
\expandafter\ifx\csname urlprefix\endcsname\relax\def\urlprefix{URL }\fi
\providecommand{\eprint}[2][]{\url{#2}}

\bibitem{Bayat2010}
Bayat A and Bose S 2010 {\em Phys. Rev. A\/} {\bf 81}(1) 012304
  \urlprefix\url{http://link.aps.org/doi/10.1103/PhysRevA.81.012304}

\bibitem{Benjamin2001}
Benjamin S~C 2001 {\em Phys. Rev. A\/} {\bf 64}(5) 054303
  \urlprefix\url{http://link.aps.org/doi/10.1103/PhysRevA.64.054303}

\bibitem{Burgarth2005_2}
Burgarth D and Bose S 2005 {\em New Journal of Physics\/} {\bf 7} 135
  \urlprefix\url{http://stacks.iop.org/1367-2630/7/i=1/a=135}

\bibitem{Burgarth2005}
Burgarth D and Bose S 2005 {\em Phys. Rev. A\/} {\bf 71}(5) 052315
  \urlprefix\url{http://link.aps.org/doi/10.1103/PhysRevA.71.052315}

\bibitem{Fitzsimons2006}
Fitzsimons J and Twamley J 2006 {\em Phys. Rev. Lett.\/} {\bf 97}(9) 090502
  \urlprefix\url{http://link.aps.org/doi/10.1103/PhysRevLett.97.090502}

\bibitem{Gong2007}
Gong J and Brumer P 2007 {\em Phys. Rev. A\/} {\bf 75}(3) 032331
  \urlprefix\url{http://link.aps.org/doi/10.1103/PhysRevA.75.032331}

\bibitem{Osborne2004}
Osborne T~J and Linden N 2004 {\em Phys. Rev. A\/} {\bf 69}(5) 052315
  \urlprefix\url{http://link.aps.org/doi/10.1103/PhysRevA.69.052315}

\bibitem{Wojcik2005}
W\'ojcik A, \L{}uczak T, Kurzy\ifmmode~\acute{n}\else \'{n}\fi{}ski P, Grudka
  A, Gdala T and Bednarska M 2005 {\em Phys. Rev. A\/} {\bf 72}(3) 034303
  \urlprefix\url{http://link.aps.org/doi/10.1103/PhysRevA.72.034303}

\bibitem{Matthias2004}
Christandl M, Datta N, Ekert A and Landahl A~J 2004 {\em Phys. Rev. Lett.\/}
  {\bf 92}(18) 187902
  \urlprefix\url{http://link.aps.org/doi/10.1103/PhysRevLett.92.187902}

\bibitem{Kosaka2008}
Kosaka H, Shigyou H, Mitsumori Y, Rikitake Y, Imamura H, Kutsuwa T, Arai K and
  Edamatsu K 2008 {\em Phys. Rev. Lett.\/} {\bf 100}(9) 096602
  \urlprefix\url{http://link.aps.org/doi/10.1103/PhysRevLett.100.096602}

\bibitem{Burgarth2007}
Burgarth D, Giovannetti V and Bose S 2007 {\em Phys. Rev. A\/} {\bf 75}(6)
  062327 \urlprefix\url{http://link.aps.org/doi/10.1103/PhysRevA.75.062327}

\bibitem{Sougato2007}
Bose S 2007 {\em Contemporary Physics\/} {\bf 48} 13--30 (\textit{Preprint}
  \eprint{http://dx.doi.org/10.1080/00107510701342313})
  \urlprefix\url{http://dx.doi.org/10.1080/00107510701342313}

\bibitem{Burgarth2009}
Burgarth D, Bose S, Bruder C and Giovannetti V 2009 {\em Phys. Rev. A\/} {\bf
  79}(6) 060305
  \urlprefix\url{http://link.aps.org/doi/10.1103/PhysRevA.79.060305}

\bibitem{Burgarth2010}
Burgarth D, Maruyama K, Murphy M, Montangero S, Calarco T, Nori F and Plenio
  M~B 2010 {\em Phys. Rev. A\/} {\bf 81}(4) 040303
  \urlprefix\url{http://link.aps.org/doi/10.1103/PhysRevA.81.040303}

\bibitem{Kempe2002}
Kempe J and Whaley K~B 2002 {\em Phys. Rev. A\/} {\bf 65}(5) 052330
  \urlprefix\url{http://link.aps.org/doi/10.1103/PhysRevA.65.052330}

\bibitem{Childs2013}
Childs A~M, Gosset D and Webb Z 2013 {\em Science\/} {\bf 339} 791--794
  (\textit{Preprint}
  \eprint{http://www.sciencemag.org/content/339/6121/791.full.pdf})
  \urlprefix\url{http://www.sciencemag.org/content/339/6121/791.abstract}

\bibitem{Janzing2007}
Janzing D 2007 {\em Phys. Rev. A\/} {\bf 75}(1) 012307
  \urlprefix\url{http://link.aps.org/doi/10.1103/PhysRevA.75.012307}

\bibitem{Bethe1931}
Bethe H 1931 {\em Zeitschrift für Physik\/} {\bf 71} 205--226 ISSN 0044-3328
  \urlprefix\url{http://dx.doi.org/10.1007/BF01341708}

\bibitem{Cotfas2012}
Cotfas N and Dragoman D 2012 {\em Journal of Physics A: Mathematical and
  Theoretical\/} {\bf 45} 425305
  \urlprefix\url{http://stacks.iop.org/1751-8121/45/i=42/a=425305}

\bibitem{NielsenChuang2011}
Nielsen M~A and Chuang I~L 2011 {\em Quantum Computation and Quantum
  Information: 10th Anniversary Edition\/} 10th ed (Cambridge University Press)
  ISBN 9781107002173

\bibitem{Loan1977}
Loan C~V 1977 {\em SIAM Journal on Numerical Analysis\/} {\bf 14} 971--981
  (\textit{Preprint} \eprint{http://dx.doi.org/10.1137/0714065})
  \urlprefix\url{http://dx.doi.org/10.1137/0714065}

\bibitem{Kitaev2002}
Kitaev A~Y, Shen A~H and Vyalyi M~N 2002 {\em Classical and Quantum Computation
  (Graduate Studies in Mathematics)\/} (Amer Mathematical Society) ISBN
  9780821832295

\bibitem{Lloyd1996}
Lloyd S 1996 {\em Science\/} {\bf 273} 1073--1078 (\textit{Preprint}
  \eprint{http://www.sciencemag.org/content/273/5278/1073.full.pdf})
  \urlprefix\url{http://www.sciencemag.org/content/273/5278/1073.abstract}

\bibitem{Ralph2011}
Ralph T and Pryde G  (\textit{Preprint} \eprint{arXiv:1103.6071})

\bibitem{Cooley1969}
Cooley J, Lewis P and Welch P 1969 {\em Audio and Electroacoustics, IEEE
  Transactions on\/} {\bf 17} 77--85 ISSN 0018-9278

\bibitem{kato}
Kato T 1995 {\em Perturbation Theory for Linear Operators (Classics in
  Mathematics)\/} 2nd ed (Springer) ISBN 9783540586616

\bibitem{Chase2008}
Chase B~A and Landahl A~J  (\textit{Preprint} \eprint{arXiv:0802.1207})

\bibitem{Cappellaro2011}
Cappellaro P, Viola L and Ramanathan C 2011 {\em Phys. Rev. A\/} {\bf 83}(3)
  032304 \urlprefix\url{http://link.aps.org/doi/10.1103/PhysRevA.83.032304}

\bibitem{Bao2015}
Bao N, Hayden P, Salton G and Thomas N 2015 {\em New Journal of Physics\/} {\bf
  17} 093028 \urlprefix\url{http://stacks.iop.org/1367-2630/17/i=9/a=093028}

\bibitem{Terhal2015}
Lloyd S and Terhal B  (\textit{Preprint} \eprint{arXiv:1509.01278})

\end{thebibliography}
\bibliographystyle{iopart-num}

\section*{Appendix}

The appendix will be structured as follows.  In section \cref{subsec_6_1} we will go over the basic theorems and definitions needed for our bounds.  These are all well known theorems already, we are merely stating them for convenience and completeness.  \cref{subsec_6_2} contains a number of useful properties of the Gaussian packets we have defined.  Namely, we show that they satisfy a discrete approximate Heisenberg relation, and that they are approximately normalized.  \cref{subsec_6_3} will cover our dispersion results.  In this section we prove some of the theorems of the paper, using mostly theorems from \cref{subsec_6_1}.  In \cref{subsec_6_4} we will leverage our dispersion result, and the well known technique of Trotterization to get our final bounds.

\subsection{Basic Theory}\label{subsec_6_1}
\subsubsection{Conventions}

$\,$\\
It is important to start with a few conventions that will be followed throughout this paper.  We will always use the standard 2-norm for vectors ($||\,\,\ket{\psi}||=\sqrt{\braket{\psi|\psi}}$), and the standard operator norm for matrices:
\begin{equation}
||A||=\max_{\ket{\psi} \in \mathbb{C}^n}\frac{||A\ket{\psi}||}{||\ket{\psi}||}
\end{equation} 
If $A$ maps a particular subspace $W$ to itself, we can define 
\begin{equation}
||A||_W=\max_{\ket{\psi} \in W} \frac{||A\ket{\psi}||}{||\ket{\psi}||}
\end{equation}
It is clear that each operator norm is compatible with the 2-norm.  In other words, 
$$
||A\ket{\psi}|| \leq ||A|| \,\,||\,\,\ket{\psi}||
$$
and if $\ket{\psi} \in W$
$$
||A\ket{\psi}|| \leq ||A||_W \,\,\,||\,\,\ket{\psi}||
$$
\subsubsection{Well Known Theorems}
$\,$ \\

Several theorems will be needed in the proofs that follow.  First, we will need a simple result bounding the sum of a function defined over a discrete set of values:
\begin{restatable}{thm}{sum_bound}\label{sum_bound}
Let f be a positive smooth function defined on $\mathbb{R}$ with one global maximum $f_{{ \rm max}}$ and no other local maximum.  Then:
\begin{equation}
\sum_{j=0}^{N-1} f(j) \leq \int_{-\infty}^\infty f(x) dx+f_{{ \rm max}}
\end{equation}
and 
\begin{equation}
\sum_{j=0}^{N-1} f(j) \geq \int_{-\infty}^\infty f(x) dx - \int_{-\infty}^{-1} f(x) dx - \int_N^\infty f(x)dx -  f_{{ \rm max}}
\end{equation}
\end{restatable}

\begin{proof}
\underline{For the upper bound}, let $k$ be the first integer such that $f(k+1)<f(k)$.  Write the sum as:
\begin{equation}
\sum_{j=0}^{k-1} f(j) +\sum_{j=k+1}^{N-1}f(j) + f(k)=\sum_{j=0}^{k-1} f(j) +\sum_{j=k}^{N-2}f(j+1)+f(k)
\end{equation}

\noindent The sum is bounded above by the function, and $f(k) \leq f_{max}$ so
\begin{equation}
\leq \int_{-\infty}^\infty f(x) dx + f_{{ \rm max}}
\end{equation}

\noindent For the \underline{lower bound}, let $k$ be defined similarly.
\begin{equation}
\sum_{j=0}^{N-1} f(j) = \sum_{j=-1}^{k-1} f(j+1)+ \sum_{j=k}^{N-1} f(j)+f_{max}-f_{max}
\end{equation}
The sum combined with $f_{max}$ must be greater than the integral of the function, minus the tails, so we have the desired inequality.
\begin{figure}[h!]
  \centering
    \includegraphics[width=0.7\textwidth]{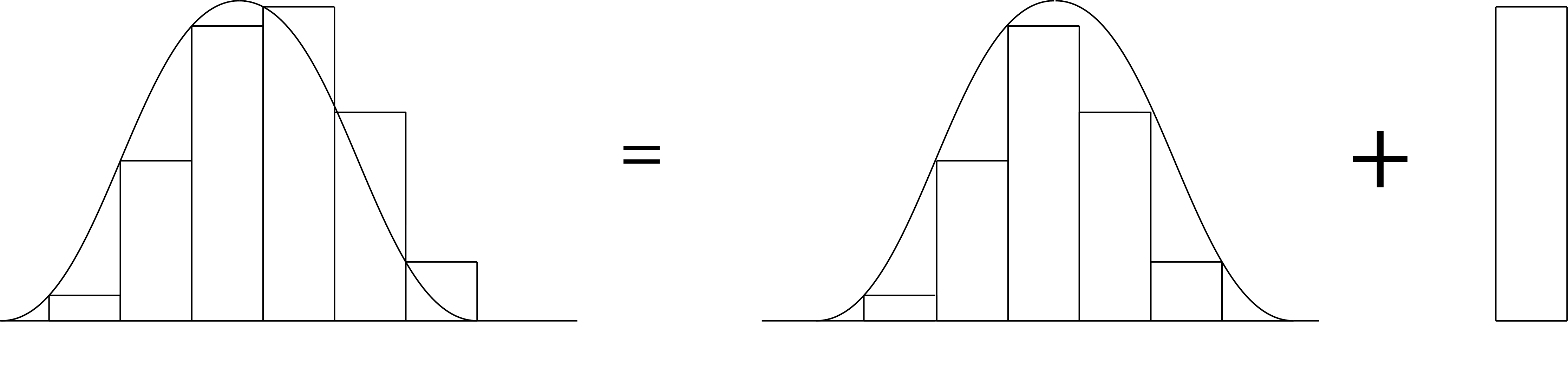}
    \caption{Upper Bound}
\end{figure}
\begin{figure}[h!]
  \centering
    \includegraphics[width=0.7\textwidth]{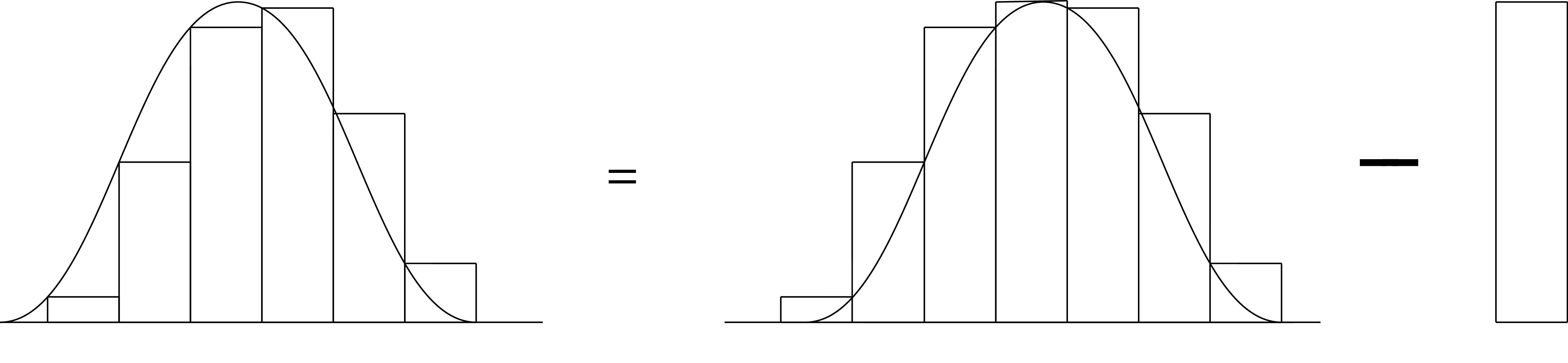}
    \caption{Lower Bound}
\end{figure}
\end{proof}

\noindent We will use this theorem to bound the dispersion present when our packets propagate around the chain in section 6.3.

The following simple theorem will be an important tool in much of our analysis.  It can be used to bound the superposition  over orthogonal quantum states, when a bound on each basis state is known.
\begin{restatable}{thm}{basis_bound} \label{basis_bound}
Let $\{\ket{x_i} \}$ be a set of states, containing p elements, with the same norm.  Let $\braket{x_i|x_j}\leq\delta_{ij} e^2$.

Then if $\sum_{j}|\beta_j|^2=1$ we have:
\begin{equation}
\left|\left| \sum_{j} \beta_j \ket{x_j}\right|\right| \leq e 
\end{equation}   

\end{restatable}

\begin{proof}
$\,$\\
\begin{equation}
\left(\sum_k \beta_k^* \bra{x_k} \right)\left(\sum_j \beta_j \ket{x_j} \right)\leq \sum_k |\beta_k|^2 e^2 =e^2
\end{equation}
\end{proof}
The most important concept for us is the so called linear propagation of error.  It is used repeatedly in almost every theorem in this paper:
\begin{restatable}[Linear Propagation of Error]{thm}{kitaev}\label{thm_kitaev}
Let $U_1, ..., U_n$ and $V_1, ..., V_n$ be unitary operators.  Then for any $\ket{\psi}$,
\begin{equation}
\left| \left| \left(\prod_{i=n}^1 U_i-\prod_{i=n}^1 V_i \right)\ket{\psi}\right| \right| \leq \sum_{j=1}^n \left| \left| (U_j-V_j)\prod_{i=j-1}^1 U_i \ket{\psi}\right| \right|
\end{equation}
\end{restatable}
\noindent Proof of this can be found in \cite{Childs2013, Kitaev2002}.  The following sensitivity result will be useful in the transient regime:
\begin{restatable}[Matrix Exponential]{thm}{matexp}
Let $A$ and $E$ be skew Hermitian complex matrices, and let t be a real scalar.  It holds that:
\begin{equation}
\left| \left|e^{(A+E)t}-e^{At} \right| \right| \leq ||E||t e^{||E||t}
\end{equation}
\end{restatable}
\begin{proof}
Found in \cite{kato} and \cite{Loan1977}
\end{proof}

We will make heavy use of the finite Fourier transform, as it provides a link between the position and momentum bases for states in the discrete Hilbert space.  Given some $N$ periodic sequence $\{A_x\,\, |\,\, x \in [0, N-1] \}$, we define its discrete finite Fourier transform as another $N$ periodic sequence $a_p=\F[A](p)$ satisfying:
\begin{equation}\label{f_trans_eq1}
a_p=\F[A](p)=\frac{1}{\sqrt{N}}\sum_{x=0}^{N-1} A_x e^{-\frac{2 \pi i p x}{N}}
\end{equation}
This transformation acts as a bijection between $x$ sequences and $p$ sequences, the inverse can be written as:
\begin{equation}\label{f_trans_eq2}
A_x= \Fi[a](x)=\frac{1}{\sqrt{N}}\sum_{p=0}^{N-1} a_p e^{\frac{2 \pi i p x}{N}}
\end{equation}

For us, this transformation is important because it provides a way to express a particular quantum state in either the position ($\ket{x}$) or the momentum basis ($\ket{p}_m$).  If we suppose that a particular quantum state has some expression in terms of either basis:
\begin{equation}
\ket{\psi}=\sum_{x=0}^{N-1} A_x \ket{x}=\sum_{p=0}^{N-1} a_p \ket{p}_m
\end{equation}
then we can see that the coefficients are related by a discrete Fourier transform.  If we take the inner product of either side with $\bra{p'}_m$ for some fixed $p'$, then we can derive \cref{f_trans_eq1}.  If we instead take $\bra{x'}$ for some fixed $x'$, then we can derive \cref{f_trans_eq2}.

\subsubsection{Definitions}
$\,$\\
Our proposed scheme takes place in the dual rail subspace $V$.  Suppose the rails are listed bit by bit, always starting with the rail corresponding to the 0th bit: $0_1 1_1 0_2 1_2... 0_m 1_m$.  The dual rail subspace $V$ is then the span of all the basis states that are consistent with the dual rail encoding.  So each pair of rails has exactly one excitation, either on the $0$ rail or on the $1$ rail.  We would write one of these states as $\ket{\psi_1^0}\ket{\psi_1^1}\ket{\psi_2^0}\ket{\psi_2^1}...\ket{\psi_m^0}\ket{\psi_m^1}$ where either $\ket{\psi_j^0}=\ket{x}$ and $\ket{\psi}=\ket{{ \rm vac}}$ or $\ket{\psi_j^0}=\ket{{ \rm vac}}$ and $\ket{\psi_j^1}=\ket{x}$.  All of our gates will maintain this subspace, and we will start in a superposition of states in this space.  Operator norms are by default $|| \,\,||_V$.

It is important to describe a relevant basis decomposition, denoted in this paper as RBD.  At some point we will have a superposition of packets inside a particular gate block.  Say we have four qubits, and we are performing an $\mathbf{X}$ gate on the first one, a $\mathbf{Z}$ gate on the second one, and a CPHASE on the last two.  This means that the packets are well localized inside a gate block that is implementing these operations on the respective qubits.  The RBD is a way to write our superposition of packets:
\begin{equation}
\ket{\psi}=\sum_i c_i \ket{\psi_i}
\end{equation}
such that each $\ket{\psi_i}$ is diagonal given the current gates.  In other words, of the current gates were extended to the whole chain, and the other gates were removed, then each $\ket{\psi_i}$ would be some set of packets that move around our qubit rings and acquire a phase.  In the case we are considering, this basis would include the encoded state $\ket{+}\ket{0}\ket{1}\ket{1}$, which would have the form:
\begin{equation}
\fl \bigg(\ket{{ \rm packet}}\ket{{ \rm vac}}+\ket{{ \rm vac}}\ket{{ \rm packet}}\bigg)\bigg(\ket{{ \rm packet}}\ket{{ \rm vac}}\bigg)\bigg(\ket{{ \rm vac}}\ket{{ \rm packet}}\bigg)\bigg(\ket{{ \rm vac}}\ket{{ \rm packet}}\bigg)
\end{equation}
The $\mathbf{X}$ gate would add some phase to this state per unit time, as well as the CPHASE gate.  

We will also need to make use of the standard computational basis decomposition (CBD).  In this case we are merely writing the state $\ket{\psi}$ as some superposition over the encoded computational basis.  We will write $\ket{\psi}=\sum_i c_i \ket{\psi_i}$ where each $\ket{\psi_i}$ is some encoded computational basis state.  In our running example, one such state would be the encoded $\ket{0}\ket{1}\ket{1}\ket{0}$ state, which would have the form:
\begin{equation}
\bigg(\ket{{ \rm packet}}\ket{{ \rm vac}}\bigg)
\bigg(\ket{{ \rm vac}}\ket{{ \rm packet}}\bigg)
\bigg(\ket{{ \rm vac}}\ket{{ \rm packet}}\bigg)
\bigg(\ket{{ \rm packet}}\ket{{ \rm vac}}\bigg)
\end{equation}

One more concept is needed before we begin our proofs.  We need to describe what it means for our set of packets to be localized to at least a distance $d$.  Suppose we have some undesirable interaction Hamiltonian $H_b$.  $H_b$ will be made up of 1-local and 2-local interactions corresponding to one and two qubit gates (CPHASE). We say the our superposition of packets ``is localized to at least a distance $d$'' if all 1-local interactions are at least a distance $d$ from our packets (mod $N$), and all undesirable 2-local interactions have at least one leg that is a distance $d$ away from our packets (mod $N$). 

Said more explicitly, suppose the packets are centered at some position $x_0 \in [0, N-1]$ on the chain, suppose that the undesirable 1-local interactions occur at some positions $A=\{x_i\} \in [0, N-1]$ on the chain, and that the undesirable 2-local interactions occur at positions $B=\{(x_i, x_j)\} \in [0, N-1]$ on the chain.  The packets ``being localized to a distance $d$'' means that for all $x_i \in A$, 
\begin{equation}
x_i-x_0 \,\,{\rm mod}\,\, N \geq d
\end{equation}
and that for all $(x_i, x_j) \in B$  
\begin{eqnarray}
(x_i-x_0) \,\, { \rm mod} \,\, N \geq d \,\,\,\,\,\,\,\,\,\,\,\, { \rm or} \,\,\,\,\,\,\,\,\,\,\,\, (x_j-x_0) \,\, { \rm mod} \,\, N \geq d
\end{eqnarray}

\subsection{Uncertainty and Normalization}\label{subsec_6_2}
\subsubsection{Approximate Heisenberg}
$\,$\\
We can switch back and forth between the momentum representation and the position representation for our packets and will show not that they are exactly the same, but that they are the same up to an exponentially small error.  Suppose at some point in the computation there is some superposition of packets in the dual rail encoding and that they are written in the CBD.  We show below that a particular basis state in the decomposition can be switched from momentum to position or vice versa, and then use \cref{basis_bound} from the last section to apply this to the whole superposition.  To prove that we can switch between position and momentum, we will first show that we can switch for a single packet on a single qubit ring.  Then, we can extend this to a tensor product of $m$ packets and finally to a superposition in the dual rail encoding.  Note that the following proof is contained in \cite{Cotfas2012}.  We provide it here for completeness, and to verify that the proof still works for the case we are interested in ($x_0 \in[0, N-1]$ a real number and $p_0=N/4$ an integer).

\begin{restatable}{thm}{packets}
Let $p_0=N/4$ (an integer), and let $x_0 \in [0, N-1]$ (not necessarily an integer).  Let $\Delta x$, $\Delta p$ be positive real numbers satisfying $2 \pi \Delta x \Delta p=N$.  It holds (up to exponentially small in N corrections) that:
\begin{eqnarray}
\fl \frac{1}{\sqrt{\Delta x\sqrt{\pi}}}\sum_{x=0}^{N-1} \sum_{\alpha=-\infty}^\infty e^{\frac{2 \pi i p_0 x}{N}}e^{-\frac{(\alpha N+x-x_0)^2}{2\Delta x^2}} \ket{x}=\nonumber \\
\frac{1}{\sqrt{\Delta p \sqrt{\pi}}} \sum_{p=0}^{N-1} \sum_{\alpha=-\infty}^\infty e^{-\frac{2 \pi i p x_0}{N}} e^{-\frac{(\alpha N+p-p_0)^2}{2\Delta p^2}}\ket{p}_m
\end{eqnarray}

\end{restatable}
\begin{proof}
Define $G_{\Delta x}: \mathbb{R} \rightarrow \mathbb{C}$ such that
\begin{equation}
G_{\Delta x}(x)=\frac{1}{\sqrt{\Delta x\sqrt{\pi}}} \sum_{\alpha=-\infty}^\infty e^{\frac{2 \pi i p_0 x}{N}}e^{-\frac{(\alpha N+x-x_0)^2}{2\Delta x^2}}
\end{equation}

$G_{\Delta x}$ has a Fourier expansion
\begin{equation}
G_{\Delta x} (x)=\sum_{l=-\infty}^\infty a_l e^{\frac{2 \pi i l x}{N}}
\end{equation} 
with
\begin{eqnarray}
a_l=\frac{1}{N} \int_0^N e^{\frac{-2 \pi i l x}{N}} \frac{1}{\sqrt{\Delta x\sqrt{\pi}}} \sum_{\alpha=-\infty}^\infty e^{\frac{2 \pi i p_0 x}{N}}e^{-\frac{(\alpha N+x-x_0)^2}{2\Delta x^2}} dx=\\
\frac{1}{N\sqrt{\Delta x\sqrt{\pi}}} \int_0^N e^{-\frac{2 \pi i l x}{N}} \sum_{\alpha=-\infty}^\infty e^{\frac{2 \pi i p_0 x}{N}}e^{-\frac{(\alpha N+x-x_0)^2}{2\Delta x^2}} dx
\end{eqnarray}
Using the substitution $t=\alpha N +x-x_0$.
\begin{eqnarray}
=\frac{1}{N\sqrt{\Delta x\sqrt{\pi}}} \sum_{\alpha=-\infty}^\infty \int_{\alpha N-x_0}^{(\alpha+1)N-x_0} e^{-\frac{2 \pi i l (t-\alpha N+x_0)}{N}} e^{\frac{2 \pi i p_0 (t-\alpha N+x_0)}{N}}e^{-\frac{t^2}{2\Delta x^2}} dt\\
\fl =\frac{1}{N\sqrt{\Delta x \sqrt{\pi}}}\sum_{\alpha=-\infty}^\infty \int_{\alpha N-x_0}^{(\alpha+1)N-x_0 } e^{-\frac{2 \pi i l t}{N}}e^{2 \pi i l \alpha } e^{-\frac{2 \pi i l x_0 }{N}}e^{\frac{2 \pi i p_0 t}{N}}e^{-2 \pi i p_0 \alpha} e^{\frac{2 \pi i p_0 x_0}{N}}e^{-\frac{t^2}{2 \Delta x^2}}dt
\end{eqnarray}
Assuming $p_0=N/4$ is an integer,
\begin{eqnarray}
=\frac{1}{N\sqrt{\Delta x\sqrt{\pi}}} e^{-\frac{2 \pi i l x_0}{N}} e^{\frac{2 \pi i p_0 x_0}{N}} \int_{-\infty}^\infty e^{\frac{2 \pi i (l-p_0) t}{N}} e^{-\frac{t^2}{2\Delta x^2}} dt
\end{eqnarray}
Up to an irrelevant global phase
\begin{equation}
=\frac{\sqrt{2 \pi} \Delta x}{N\sqrt{\Delta x\sqrt{\pi}}}  e^{-\frac{2 \pi i l x_0}{N}} e^{-\frac{(l-p_0)^2}{2 \Delta p^2}}= \frac{1}{\sqrt{N}}\frac{1}{\sqrt{\Delta p \sqrt{\pi}}}e^{\frac{-2 \pi i l x_0}{N}} e^{-\frac{(l-p_0)^2}{2 \Delta p^2}}
\end{equation}
So, $G_{\Delta x}(x)$ has the form

\begin{equation}
\frac{1}{\sqrt{N}}\frac{1}{\sqrt{\Delta p \sqrt{\pi}}} \sum_{l=-\infty}^\infty e^{\frac{2 \pi i l x}{N}}  e^{\frac{-2 \pi i l x_0}{N}} e^{-\frac{(l-p_0)^2}{2 \Delta p^2}}
\end{equation}
when $x$ is an integer, this can be rewritten as:
\begin{equation}
G_{\Delta x}(x)=\frac{1}{\sqrt{N}} \sum_{p=0}^{N-1} e^{\frac{2 \pi i p x}{N}} \frac{1}{\sqrt{\Delta p\sqrt{\pi}}}\sum_{\alpha=-\infty}^\infty   e^{-2 \pi i \alpha x_0} e^{-\frac{2 \pi i p x_0}{N}} e^{-\frac{(\alpha N+p-p_0)^2}{2 \Delta p^2}}
\end{equation}
We have found a p sequence for which the inverse Fourier transform is the x sequence.   They must be Fourier pairs.

Note that we have proven 
\begin{eqnarray}
\frac{1}{\sqrt{\Delta x\sqrt{\pi}}}\sum_{x=0}^{N-1} \sum_{\alpha=-\infty}^\infty e^{\frac{2 \pi i p_0 x}{N}}e^{-\frac{(\alpha N+x-x_0)^2}{2\Delta x^2}} \ket{x}=\\
\frac{1}{\sqrt{\Delta p \sqrt{\pi}}} \sum_{p=0}^{N-1} \sum_{\alpha=-\infty}^\infty e^{-2 \pi i \alpha x_0} e^{-\frac{2 \pi i p x_0}{N}} e^{-\frac{(\alpha N+p-p_0)^2}{2\Delta p^2}}\ket{p}_m
\end{eqnarray}
The state on the right hand side is exponentially close to the state given in the statement of the theorem\footnote{Note that if $x_0$ was an integer, we would have exact equality.  However, we need the proof to be valid when $x_0$ is not an integer, since we will move our packets in very small steps.}.  The only difference is the extra phase $e^{-2 \pi i \alpha x_0}$, but the $|\alpha|\geq 1$ contributions are exponentially small since $p_0=N/4$, and $\Delta p \sim N^{1-\varepsilon}$.  For more explicit examples see our normalization proof below.
\end{proof}

With the analysis of a single packet on a single rail in hand, we can show the analogous result for a tensor product on some set of rails.  Let $\ket{\psi^x}$ be the position packet, and let $\ket{\psi^p}$ be the momentum packet, and suppose we are interested in the tensor of $m$ packets: $\ket{\psi_1^x}\ket{\psi_2^x}...\ket{\psi_m^x}$.  We know that $\ket{\psi^x}=\ket{\psi^p}+\ket{\varepsilon}$ for $\varepsilon=\braket{\varepsilon|\varepsilon}$ exponentially small in $N$.  We are interested in bounding:
\begin{eqnarray}
\left|\left|\ket{\psi_1^x}\ket{\psi_2^x}...\ket{\psi_m^x}-\ket{\psi_1^p}\ket{\psi_2^p}...\ket{\psi_m^p}\right|\right|=\\
\left|\left| (\ket{\psi_1^p}+\ket{\varepsilon})(\ket{\psi_2^p}+\ket{\varepsilon})...(\ket{\psi_m^p}+\ket{\varepsilon})-\ket{\psi_1^p}\ket{\psi_2^p}...\ket{\psi_m^p}  \right|\right|
\end{eqnarray}
\renewcommand{\kbldelim}{(}
\renewcommand{\kbrdelim}{)}
Using the triangle inequality,
\begin{equation}
\leq \sum_{k=1}^m \kbordermatrix{
    \mbox{} & \mbox{}\\
    \mbox{} & m \\ 
   \mbox{} & k
}  \sqrt{\braket{\psi_p|\psi_p}^{m-k} \varepsilon^k}
\end{equation}  
Given the results of the following section, this sum is exponentially small.

So far it has been shown that the tensor product of $m$ packets can be switched from momentum to position or vice versa while only incurring an exponentially small error.  We need to extend this to the case where we have a superposition over sets of packets in the dual rail encoding.  Suppose we have some state written in the CBD:
\begin{equation}
\ket{\psi}=\sum_j c_j \ket{\psi_j}
\end{equation}
where each $\ket{\psi}$ is the tensor product of $m$ packets on their appropriate rails and $m$ vacuum states.  $\ket{\psi_j^p}$ is used to denote a tensor of packets in the momentum representation, and $\ket{\psi_j^x}$ denotes a tensor of these packets in the position representation.  We wish to bound the difference
\begin{equation}
\left|\left|\sum_i c_i \left(\ket{\psi_i^x}-\ket{\psi_i^p} \right)\right|\right|=\left|\left|\sum_i c_i \ket{\eta^i}\right|\right|
\end{equation}
Observe that $\braket{\eta_i|\eta_j} \leq \delta_{ij} \braket{\eta_i|\eta_i}$, so we can use \cref{basis_bound}, and bound the sum with a bound on each state in the sum.  We have already shown $\braket{\eta_i|\eta_i}$ is exponentially small with $N$, so the sum must also be exponentially small with $N$.  So, we have that a superposition of Gaussian packets in the dual rail encoding written in the position basis is exponentially close to that same Gaussian written in the momentum basis assuming $2 \pi\Delta x \Delta p= N$


\subsubsection{Approximate Normalization}

Next we demonstrate that our packets so defined are very close to being normalized in the limit of large $N$.  This is important, since the real physical state will be normalized, and we will work mostly with slightly unnormalized states.  First we will show that a single packet on a single rail is very nearly normalized, next that a tensor product of such packets is nearly normalized, and finally that a superposition must also be nearly normalized.

\begin{restatable}[Approximate Normalization]{thm}{approx_normal}\label{approx_normal}
Define 
\begin{equation}
\ket{\psi}=\frac{1}{\sqrt{\Delta p \sqrt{\pi}}}\sum_{p=0}^{N-1} \sum_{\alpha=-\infty}^\infty e^{-\frac{2 \pi i x_0 p }{N}} e^{-\frac{(\alpha N+p-p_0)^2}{2 \Delta p^2}} \ket{p}
\end{equation}
We have that,
\begin{equation}
|\braket{\psi|\psi}| = 1+ O \left(\frac{1}{\Delta p}\right)
\end{equation}
\end{restatable}

\begin{proof}
$\,$\\
\begin{eqnarray}
\fl\braket{\psi | \psi}= \nonumber \\
\fl \frac{1}{\Delta p \sqrt{\pi}}\left(\sum_{p'=0}^{N-1}\sum_{\alpha'=-\infty}^\infty e^{\frac{2 \pi i x_0 p }{N}} e^{-\frac{(\alpha' N+p-p_0)^2}{2 \Delta p^2}} \bra{p'} \right) \left( \sum_{p=0}^{N-1}\sum_{\alpha=-\infty}^\infty e^{-\frac{2 \pi i x_0 p }{N}} e^{-\frac{(\alpha N+p-p_0)^2}{2 \Delta p^2}} \ket{p}  \right)
\end{eqnarray}
\begin{equation}
= \frac{1}{\Delta p \sqrt{\pi}} \sum_{\alpha, \alpha'=-\infty}^\infty \sum_{p=0}^{N-1} e^{-\frac{(\alpha N+p-p_0)^2}{2\Delta p^2}} e^{-\frac{(\alpha' N+p-p_0)^2}{2\Delta p^2}}
\end{equation}
\begin{eqnarray} 
\fl =\frac{1}{\Delta p \sqrt{\pi}} \Bigg( \sum_{p=0}^{N-1} e^{-\frac{(p-p_0)^2}{\Delta p^2}} +2\sum_{p=0}^{N-1} e^{-\frac{(N+p-p_0)^2}{2\Delta p^2}} e^{-\frac{(p-p_0)^2}{2\Delta p^2}} + \nonumber \\
 \sum_{|\alpha|, |\alpha'|\geq 1}\sum_{p=0}^{N-1}  e^{-\frac{(\alpha N+p-p_0)^2}{2\Delta p^2}} e^{-\frac{(\alpha' N + p-p_0)^2}{2 \Delta p^2}} \Bigg) 
\end{eqnarray} 
We assume that $p_0=N/4$.  Then, 
\begin{eqnarray}
\fl \leq \frac{1}{\Delta p \sqrt{\pi}} \Bigg( \sum_{p=0}^{N-1} e^{-\frac{(p-p_0)^2}{\Delta p^2}} +2 e^{-\frac{\frac{9}{16}N^2}{2\Delta p^2}}\sum_{p=0}^{N-1} e^{-\frac{(p-p_0)^2}{2 \Delta p^2}} +\nonumber \\
\sum_{|\alpha|, |\alpha'|\geq 1}\sum_{p=0}^{N-1}  e^{-\frac{(\alpha N+p-p_0)^2}{2\Delta p^2}} e^{-\frac{(\alpha' N + p-p_0)^2}{2 \Delta p^2}} \Bigg)
\end{eqnarray}
Assuming $\Delta p=N^{1-\varepsilon}$ the second and third terms are exponentially small, and can be ignored.  Dropping the small terms and applying \cref{sum_bound}, we get to leading order
\begin{equation}
\leq \frac{1}{\Delta p \sqrt{\pi}}\left(1+\int_{-\infty}^\infty e^{-\frac{(p-p_0)^2}{\Delta p^2}} dp \right)\leq 1 + \frac{1}{\Delta p \sqrt{\pi}}
\end{equation}

To derive the lower bound, we start from equation $(67)$, drop the small terms and apply \cref{sum_bound} to get:
\begin{equation}
\braket{\psi | \psi } \geq \frac{1}{\Delta p \sqrt{\pi}} \left( \Delta p \sqrt{\pi} - \int_{-\infty}^{-1}e^{-\frac{(p-p_0)^2}{\Delta p^2}} dp + \int_{N}^\infty e^{-\frac{(p-p_0)^2}{\Delta p^2}} dp -1\right)
\end{equation}
We can upper bound the two integrals (which corresponds to  a lower bound on the negative integrals) and show that they are exponentially small with the observation that $e^{-x^2} \leq e^{-x}$.  The final asymptotic bound is
\begin{equation}
\braket{\psi | \psi} \geq 1-\frac{1}{\Delta p \sqrt{\pi}}
\end{equation}
\end{proof}

So, now suppose we have a tensor of packets in the momentum representation.  We can write the normalization constant (up to exponentially small corrections) as:
\begin{eqnarray}
\fl \sqrt{\braket{\psi_1^p|\psi_1^p}\braket{\psi_2^p|\psi_2^p}...\braket{\psi_m^p|\psi_m^p}}\leq \sqrt{\left(1+\frac{1}{\Delta p}\right)^m}=\Bigg[1+\sum_{k=1}^m \kbordermatrix{
    \mbox{} & \mbox{}\\
    \mbox{} & m \\ 
   \mbox{} & k
} \left(\frac{1}{\Delta p} \right)^k\Bigg]^{1/2} =\nonumber \\
\fl \sqrt{1+ O\left(\frac{m}{\Delta p} \right)}
\end{eqnarray}
In our final scaling, $\frac{m}{\Delta p} \approx \frac{1}{N^{1/3}}$ so we expect tensor products of these packets to be very near normalized.  

Just as in the last section, one more step is needed.  We need to check that a superposition over tensors of these packets will be nearly normalized.  Let us write the state in the CBD $\ket{\psi}=\sum_i c_i \ket{\psi^i}$.  Then, if $\braket{\psi^i|\psi^i}$ is the same for all $i$ and $\sum_i |c_i|^2=1$
\begin{equation}
\braket{\psi|\psi}=\braket{\psi^i|\psi^i}
\end{equation}
so the superposition is off from normalization by the same factor that the tensor of packets is.

There is one more important observation one should keep in mind as we continue with the proofs.  We need to observe that the ``translation unitary'' is well defined.  If we look at the states 

\begin{equation}
\ket{\phi_1}=\frac{1}{\sqrt{\Delta p \sqrt{\pi}}}\sum_{\alpha, p} e^{-\frac{2 \pi i p x_0}{N}} e^{-\frac{(\alpha N+p-p_0)^2}{2\Delta p^2}} \ket{p}_m
\end{equation}
and
\begin{equation}
\ket{\phi_2}=\frac{1}{\sqrt{\Delta p \sqrt{\pi}}}\sum_{\alpha, p} e^{-\frac{2 \pi i p (x_0+2t)}{N}} e^{-\frac{(\alpha N+p-p_0)^2}{2\Delta p^2}} \ket{p}_m
\end{equation}
it is not hard to see that $\braket{\phi_1|\phi_1}=\braket{\phi_2|\phi_2}$.  These states are un-normalized, but have the same vector norm.  We define $U_{2t}^{x_0}$ to be the unitary that moves $\ket{\phi_1}$ to $\ket{\phi_2}$ and maps the all spin down vector to itself.  One can then expand this unitary to be a unitary acting on the full set of states, it does not matter for us what the unitary does to the rest of the states.  In most of the paper the $x_0$ will not be explicitly written.  Note that $t$ can very well be a small real number, not necessarily an integer mod $N$.

\subsection{Dispersion and Transience}\label{subsec_6_3}

Here we will derive our dispersion bounds.  The essential idea is to write our state in the momentum basis.  In this  basis, the ring Hamiltonian is diagonal, so the exact time evolution of the quantum state can be described and the energy function can be taylor expanded.  The first order term will lead to translation of the packets, and the higher order terms will constitute some remainder whose effect we need to bound.

We make use of the same type of argument as the previous sections.  A single packet on a single rail propagates around the chain while only picking up a (polynomially) small error, which is extended to multiple chains in the dual rail encoding.  

\disp

\begin{proof}
Let $\ket{\psi}$ be a packet centered at $x_0$ in x space and $N/4$ in p space.  Time evolving, we get:
\begin{equation}
e^{-iH_{{ \rm ring}}^1t} \ket{\psi}=\frac{1}{\sqrt{\Delta p \sqrt{\pi}}} \sum_{p=0}^{N-1}\sum_{\alpha=-\infty}^\infty e^{-\frac{(\alpha N +p-p_0)^2}{2 \Delta p^2}} e^{-\frac{i 2 \pi p x_0}{N}}e^{-2 i t \cos \left[\frac{2 \pi p}{N} \right]}\ket{p}
\end{equation}
We let $\widetilde{p}=p-p_0$, and expand the energy function around $p_0=N/4$ to obtain:
\begin{equation}
\cos \left[\frac{2 \pi p} {N} \right]= \frac{2 \pi }{N} \widetilde{p}-\left(\frac{2 \pi}{N} \right)^3 \frac{1}{3!} \widetilde{p}^3+\left(\frac{2 \pi}{N} \right)^5 \frac{1}{5!} \widetilde{p}^5-...=\frac{2 \pi }{N} \widetilde{p} +R(p)
\end{equation}
The state becomes:
\begin{equation}
\fl e^{-iHt} \ket{\psi}=e^{-\frac{4 \pi i t p_0 }{N}}\frac{1}{\sqrt{\Delta p \sqrt{\pi}}} \sum_{p=0}^{N-1}\sum_{\alpha=-\infty}^\infty e^{-\frac{(\alpha N +p-p_0)^2}{2 \Delta p^2}} e^{-\frac{i 2 \pi p (x_0+2t)}{N}}e^{-2it R(p)}\ket{p}
\end{equation}
If it were not for the remainder term $R(p)$, this would exactly be the translated state we are interested in $\ket{\psi_{2t}}$ (up to an irrelevant global phase).  Subtracting the two gives
\begin{equation}
\fl e^{-iHt} \ket{\psi}-\ket{\psi_{2t}}=\frac{1}{\sqrt{\Delta p \sqrt{\pi}}}\sum_{p=0}^{N-1} \sum_{\alpha=-\infty}^\infty e^{-\frac{(\alpha N +p-p_0)^2}{2 \Delta p^2}} e^{-\frac{i 2 \pi p (x_0+2t)}{N}}(e^{-2it R(p)}-1)\ket{p}
\end{equation}
We want to bound the size of this state and show that it is (polynomially) small.  The terms with $|\alpha| \geq 1$ are negligible, since they lead to exponentially small contributions.

Dropping these terms and calculating the norm of this state, we get:
\begin{equation}
\left| \left| U \ket{\psi}-\ket{\psi_t} \right| \right|^2=\frac{1}{\Delta p \sqrt{\pi}} \sum_{p=0}^{N-1}e^{-\frac{(p-p_0)^2}{ \Delta p^2}} (e^{2itR(p)}-1)(e^{-2 i t R(p)}-1)
\end{equation}
\begin{eqnarray}
=\frac{1}{\Delta p \sqrt{\pi}} \sum_{p=0}^{N-1}e^{-\frac{(p-p_0)^2}{ \Delta p^2}} \sin \left[t R(p) \right]^2 \leq \nonumber \\
\frac{1}{\Delta p \sqrt{\pi}}\left( \sum_{p=0}^{\lceil N/2 \rceil} e^{-\frac{(p-p_0)^2}{\Delta p^2}} \sin \left[t R(p) \right]^2 +\sum_{p=\lceil N/2 \rceil+1}^{N-1} e^{-\frac{(p-p_0)^2}{\Delta p^2}} \right)
\end{eqnarray}
Recall that $p_0=N/4$, so each term in the second sum can be bounded by $e^{-\frac{(p-p_0)^2}{\Delta p^2}}$ evaluated at $N/2$.  This combined with the fact that $\sin(\theta)^2 \leq \theta^2$ gives:
\begin{equation}
\leq \frac{1}{\Delta p \sqrt{\pi}} \sum_{p=0}^{N/2}e^{-\frac{(p-p_0)^2}{ \Delta p^2}} \left[t R(p) \right]^2 +\frac{N}{2\Delta p \sqrt{\pi}} e^{-\frac{(N/4)^2}{\Delta p^2}}
\end{equation}
Assuming $\Delta p$ is  $O(N^{1-\varepsilon})$ the second term can be dropped (again it is exponentially small).  For the first sum, observe that $R(p)$ is an alternating convergent series whose magnitude (term by term) is decreasing (We split up the sum $\sum_{p=0}^{N-1}$ into two sums so that we can guarantee that the series is term by term decreasing).  So it must be smaller in magnitude than its first element.
\begin{equation}
\leq \frac{1}{\Delta p \sqrt{\pi}} \sum_{p=0}^{N/2}e^{-\frac{(p-p_0)^2}{ \Delta p^2}} \left[t \left(\frac{2 \pi}{N} \right)^3 \frac{1}{3!} \widetilde{p}^3  \right]^2
\end{equation}
We are summing a positive function.  Applying \cref{sum_bound} and dropping unimportant constants gives:

\begin{equation}
\leq \frac{t^2}{\Delta p N^6 }\left[ \int_{-\infty}^{\infty} e^{-\frac{\widetilde{p}^2}{\Delta p ^2}} \widetilde{p}^6 dp +\Delta p^6 \right]= \frac{t^2}{\Delta p N^6 }\left[\Delta p^7 +\Delta p^6 \right]
\end{equation}
to leading order
\begin{equation}
=\frac{t^2 \Delta p^6}{N^6}
\end{equation}

\end{proof}

Each packet picks up a phase of $e^{-\frac{4 \pi i \Delta t p_0}{N}}$ over a time interval $\Delta t$.  In the dual rail encoding each basis state is the tensor of $m$ packets, so every state picks up a phase of $e^{-\frac{4 \pi i m \Delta  t p_0}{N}}$.  Overall phase does not effect measurement statistics, so this can be ignored.

The corollary that extends this result to multiple rails is a simple consequence of \cref{thm_kitaev}:
\mult
\begin{proof}
Letting $H_j$ be the nearest neighbor ring Hamiltonian for the jth rail, the full Hamiltonian decomposes as $H_{{ \rm rings}}^m=H_1+...+H_m$.  Since $[H_i, H_j]=0$, $e^{-i(H_1+...+H_m)t}=e^{-iH_1 t} e^{-i H_2 t}...e^{-iH_m t}$.  Let $U_j$ translate ring j a distance $2t$ while acting as identity on the others.  Then, $U_{2t}=U_1 U_2...U_m$.  If we let $V_j=e^{-i H_j t}$ then directly applying \cref{thm:kitaev} we have 
\begin{equation}
\left| \left|\left(e^{-iH_{{ \rm rings}}^mt}-U_1U_2...U_m\right)\ket{\psi} \right| \right|=O\left(\frac{mt \Delta p^3}{N^3} \right)
\end{equation}
\end{proof}

The same kind of analysis can  be applied to bound the dispersion when we have a superposition in the dual rail encoding.  Suppose we have such a superposition written in the CBD: $\sum_i c_i \ket{\psi^i}$.  We are interested in bounding:
\begin{equation}
\sum_i c_i \left(e^{-i H_{{ \rm rings}}^{2m} t}-U_{2t} \right)\ket{\psi^i}=\sum_i c_i \ket{\eta^i}
\end{equation}
Suppose a particular $\ket{\psi^i}$ is supported on some fixed set of rails.  Each ring Hamiltonian $H_{{ \rm ring}}^1$ moves excitation to different partitions on the rail, but does not move excitations between rails.  It follows that $e^{-i H_{{ \rm ring}}^{2m}t}\ket{\psi_i}$ is supported on the same set of rails as $\ket{\psi_i}$.  Clearly the same property holds for $U_{2t}$, so we can conclude that $\braket{\eta^i|\eta^j}=0$ when $i \neq j$.  \cref{basis_bound} implies that a superposition can be bounded just as a single tensor of packets.  

So, we have shown that we can move a superposition of packets around the chain, while accruing error at most $O\left(\frac{mt\Delta p^3}{N^3} \right)$.  Notice that we could have also included a Hamiltonian which commutes with all of the rail Hamiltonians.  Say we had a set of gates that were extended to the whole chain.  Each of the gates would commute with the ring Hamiltonian, and, as long as they act on separate rings, they would also commute with each other.  In this case the resulting time evolution would just be translation and phase added to the packets.

\subsubsection{Transient Regime}
$\,$\\
As described in the paper, we need bounds on the so called ``transient regime'', where our packets are not well localized to any particular gate block.  We can leverage the fact that the the interactions are weak, so if the transient regime is small we expect to see very little effect on the packets.  For this, we use a result that bounds the matrix exponential of a complex Hermitian matrix when we add another small complex Hermitian matrix to it.  Suppose we have two complex matrices $A$ and $E$ and suppose that $E$ is small (in operator norm).  Then, the result is that $e^{(A+E)t}$ is very nearly $e^{At}$, up to an error of $O(||E||t)$.  We will need to define $U_{2t}^p$.  This is simply the translation unitary, with some added phase on the relevant basis states, depending on the current gate.

Now we will prove our transient bound:
\begin{restatable}[Bound for Transient Regime]{cor2}{transregime} \label{transregime2}
Let $H_{{ \rm ring}}^{2m}$ be the ring Hamiltonian, and let $\hat{H}$ be the gate Hamiltonians.  Suppose $\Phi$ is the maximum gate strength.  Then if $\ket{\psi}$ is a superposition of packets it holds that:
\begin{equation}
\left| \left|\left( e^{-i(H_{{ \rm ring}}^{2m}+\hat{H})t}-U_{2t}^p\right)\ket{\psi} \right| \right| =O\left(mt\Phi++\frac{m t \Delta p^3}{N^3}\right)
\end{equation}
\end{restatable}

\begin{proof}
$H_{{ \rm ring}}^{2m}$ and $\hat{H}$ map $V \rightarrow V$.  Since $\ket{\psi}$ is nearly normalized, it holds upto a constant multiple that:
\begin{equation}\label{eq_724}
\left|\left|\left(e^{-i(H_{{ \rm ring}}^{2m}+\hat{H})t}-e^{-iH_{{ \rm ring}}^{2m} t} \right)\ket{\psi} \right| \right| \leq \left|\left|e^{-i(H_{{ \rm ring}}^{2m}+\hat{H})t}-e^{-iH_{{ \rm ring}}^{2m}t}  \right| \right|_V
\end{equation}
applying \cref{loan_thing} gives $\leq ||\hat{H}||_V t e^{||\hat{H}||_V t}$

In order to evaluate $||\hat{H}||_V$, consider the standard basis for V.  These are all spin configurations with exactly one excitation per every set of two rails corresponding to a qubit.  $\hat{H}$ is the sum of 1-local and 2-local terms.  On any particular basis state $\ket{x} \in V$, at most $m$ of these terms act on $\ket{x}$ to produce something that is nonzero.  Further, if two of these computational basis states are orthogonal before action by $\hat{H}$, they will still be orthogonal after action by $\hat{H}$.  We have therefore met the conditions required to apply \cref{basis_bound} (For two basis states $\ket{x}$ and $\ket{x'}$ we have that $\braket{x'|\hat{H}^{\dagger} \hat{H} | x} \leq (m \Phi)^2 \delta_{x, x'}$) and can thus conclude: $||\hat{H}|| \leq m \Phi$.

Applying \cref{loan_thing} to \cref{eq_724}, up to dispersion, we can ignore the effect of the packets entering the gates up to an error of $||\hat{H}t||_V=O(\Phi mt)$.  Applying our previous result on dispersion, \cref{disp_multiring}, we get an additional error of 
\begin{equation}
\frac{mt\Delta p^3}{N^3}=\frac{mt}{N^{1+3 \varepsilon}}
\end{equation}

Now we need to show that the phase we are ignoring (we are skipping over a portion of the gate with width $t$) is small.  Let us write $\ket{\psi}$ in the RBD.  
The lost phase is on the order of at most $e^{imt\Phi}$ for each of these states, and our error vectors are orthogonal so we can apply \cref{basis_bound}
\begin{equation}
||(e^{imt\Phi}-1)||=\left|\sin mt\Phi \right| \approx mt\Phi
\end{equation}
The lost phase is of the same order as the error calculated using \cref{loan_thing}.
\end{proof}

\subsection{Trotterization}\label{subsec_6_4}

We will first show that a ``far away'' Hamiltonian produces an exponentially small vector when acting on a packet:

\begin{restatable}{thm}{trans}\label{small_ham}
Let $\ket{\psi}$ be a superposition of packets in the dual rail encoding on $m$ rails located at some fixed position $x_0$.  Let $H$ be an ``undesireable'' Hamiltonian containing a polynomial number of interaction terms that are all at least a distance $d$ away.  Suppose that $H$ is made up of gate interactions described in the paper, and that all these interactions have polynomial strength.  Then, we have
\begin{equation}
\left|\left|H \ket{\psi}\right|\right|=O\left(poly(N) e^{-\frac{d^2}{2\Delta x^2}} \right)
\end{equation}
\end{restatable}
\begin{proof}
Let $H=\sum_j H_j$.  Then
\begin{equation}
||H\ket{\psi}|| \leq \sum_j ||H_j \ket{\psi}||
\end{equation}
There are three cases to consider.  $H_j$ must be from an X gate ($\mathbf{X}\mathbf{X} +\mathbf{Y}\mathbf{Y}$ between 0 and 1 rails), a Z gate ($\mathbb{I}+\mathbf{Z}$ on a particular $1$ rail), or a CPHASE gate ($(\mathbb{I}+\mathbf{Z})\otimes (\mathbb{I}+\mathbf{Z})$ between 1 rails).  If $H_j$ comes from an X gate, say it is on the first qubit, then write the superposition in the RBD.  Write the state as 
\begin{equation}
\sum_{i\in [0, 2^{m-1}-1]} c_{+, i} \ket{+, i}+c_{-, i}\ket{-, i}
\end{equation}
where $\{+, -, i\}$ stands for the encodings of these states into packets.

It is easy to check that $\braket{+, i|H_j^\dagger H_j|-, i}=0$ so we can apply \cref{basis_bound} to the states $\{H_j \ket{\pm, i}\}$ and bound the superposition by a bound on each basis state.  We can write a basis state explicitly as
\begin{equation}
\left[\frac{1}{\sqrt{\Delta x \sqrt{\pi}}}\sum_{\alpha, x}e^{-\frac{(\alpha N+x-x_0)^2}{2\Delta x^2}}e^{\frac{2 \pi i p_0 x }{N}}\ket{\widetilde{x}}\right]\otimes ...
\end{equation}
where $\ket{\widetilde{x}}=\frac{\ket{x}\ket{{ \rm vac}} \pm \ket{{ \rm vac}}\ket{x}}{\sqrt{2}}$

Since $\ket{\psi^i}$ is a tensor product of packets which are all nearly normalized,
\begin{equation}
||H_j \ket{\psi^i}||=O\left(\left|\left|\frac{1}{\sqrt{\Delta x \sqrt{\pi}}}\sum_{\alpha, x}e^{-\frac{(\alpha N+x-x_0)^2}{2\Delta x^2}}e^{\frac{2 \pi i p_0 x }{N}}H_j\ket{\widetilde{x}} \right)\right|\right|
\end{equation}

$H_j\ket{\widetilde{x}}$ will be $0$ unless $(x-x_0)$ mod $N$ $\geq d$.  So, up to exponentially small contributions (dropping $|\alpha| \geq 1$), this is
\begin{equation}
O\left( \frac{\Phi N}{\sqrt{\Delta x}}e^{-\frac{d^2}{2\Delta x^2}}\right)
\end{equation}
The same analysis applies if $H_j$ is from a Z or CPHASE gate.  In the Z gate case $\ket{\widetilde{x}}$ is either $\ket{x}\ket{{ \rm vac}}$ or $\ket{{ \rm vac}}\ket{x}$.  In the CPHASE case, (suppose the CPHASE gate is on the first two qubits), we have to write down the tensor product of two packets:
\begin{equation}
\fl \left[\frac{1}{\sqrt{\Delta x \sqrt{\pi}}}\sum_{\alpha, x}e^{-\frac{(\alpha N+x-x_0)^2}{2\Delta x^2}}e^{\frac{2 \pi i p_0 x }{N}}\ket{\widetilde{x}_1}\right]\left[\frac{1}{\sqrt{\Delta x \sqrt{\pi}}}\sum_{\alpha, x}e^{-\frac{(\alpha N+x-x_0)^2}{2\Delta x^2}}e^{\frac{2 \pi i p_0 x }{N}}\ket{\widetilde{x}_2}\right]...
\end{equation}
where each $\ket{\widetilde{x}_i}$ is $\ket{x}\ket{{ \rm vac}}$ or $\ket{{ \rm vac}}\ket{x}$.  Then we have that $H_j\ket{\widetilde{x}_1} \ket{\widetilde{x}_2}$ is zero unless $\widetilde{x}_1$ or $\widetilde{x}_2$ is at least a distance $d$ away.  The resulting upper bound is the same, $O\left(poly(N) e^{-\frac{d^2}{2\Delta x^2}} \right)$ although this time the polynomial $\sim N^2 \Phi$
\end{proof}

\begin{wrapfigure}{r}{0.5\textwidth}
  \begin{center}
    \includegraphics[width=0.48\textwidth]{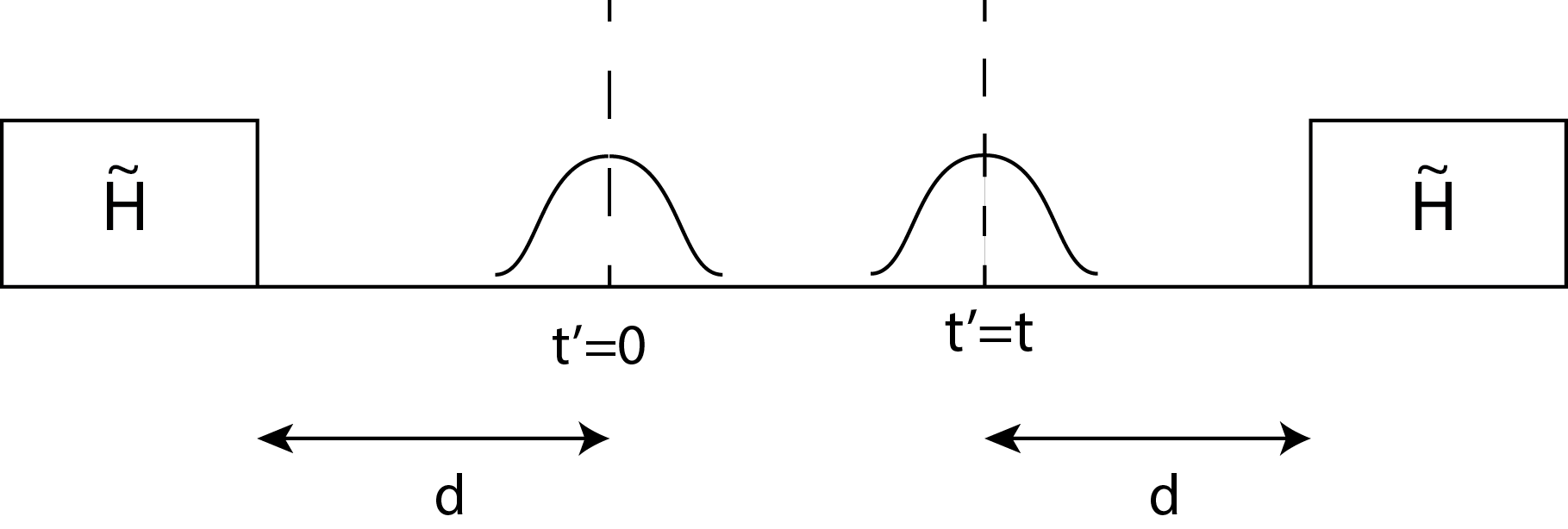}
  \end{center}
  \caption{Packets are localized to a distance d for the whole time evolution if they do not see any interactions in $\widetilde{H}$}\label{local_fig}
\end{wrapfigure}

Now we can proceed to our bounds based on Trotterization\cite{Lloyd1996}.  In the previous sections we covered how the ring Hamiltonians with some gate Hamiltonian on top can in effect propagate the packets and place some uniform phase on them.  We assumed that the gate Hamiltonians spanned the entire qubit ring.  Clearly this is not the setup we have in mind.  We want to have localized gates that implement encoded unitaries on our spin chains, and we want to send our pulses through a number of these gates to implement some quantum circuit.  In this section we will show that our local gates very nearly approximate gates that span the entire chain assuming that our packets are well localized.  We will break up the proof into two separate theorems.

First we will assume that the packets are localized inside a particular gate block, and that there are no other gate blocks.  We will then use this result to treat the more important case, where the packets are localized inside some gate block, and we are in the process of evaluating a quantum circuit.

Suppose we have some superposition of packets located at some position $x_0$, and suppose they are localized to at least a distance $d$.  We will denote the Hamiltonian corresponding only to the current local gates as the ``current Hamiltonian'' $H_{{ \rm current}}$.

In the following theorem we will make the assumption that the packets remain localized to a distance $d$ for $0 \leq 2 t' \leq 2 t$.  In the proof, we will have some undesirable interactions $\widetilde{H}$.  This means that all of the interactions in $\widetilde{H}$ are at least a distance $d$ away from $x_0+2t'$ where $0 \leq 2 t' \leq 2 t$.  While the packet is moving, each of the ``bad'' interactions remains at least a distance $d$ away.  Let us be more concrete for the CPHASE case, since this may be the most confusing part of the result.

\begin{wrapfigure}{r}{0.5\textwidth}
  \begin{center}
    \includegraphics[width=0.48\textwidth]{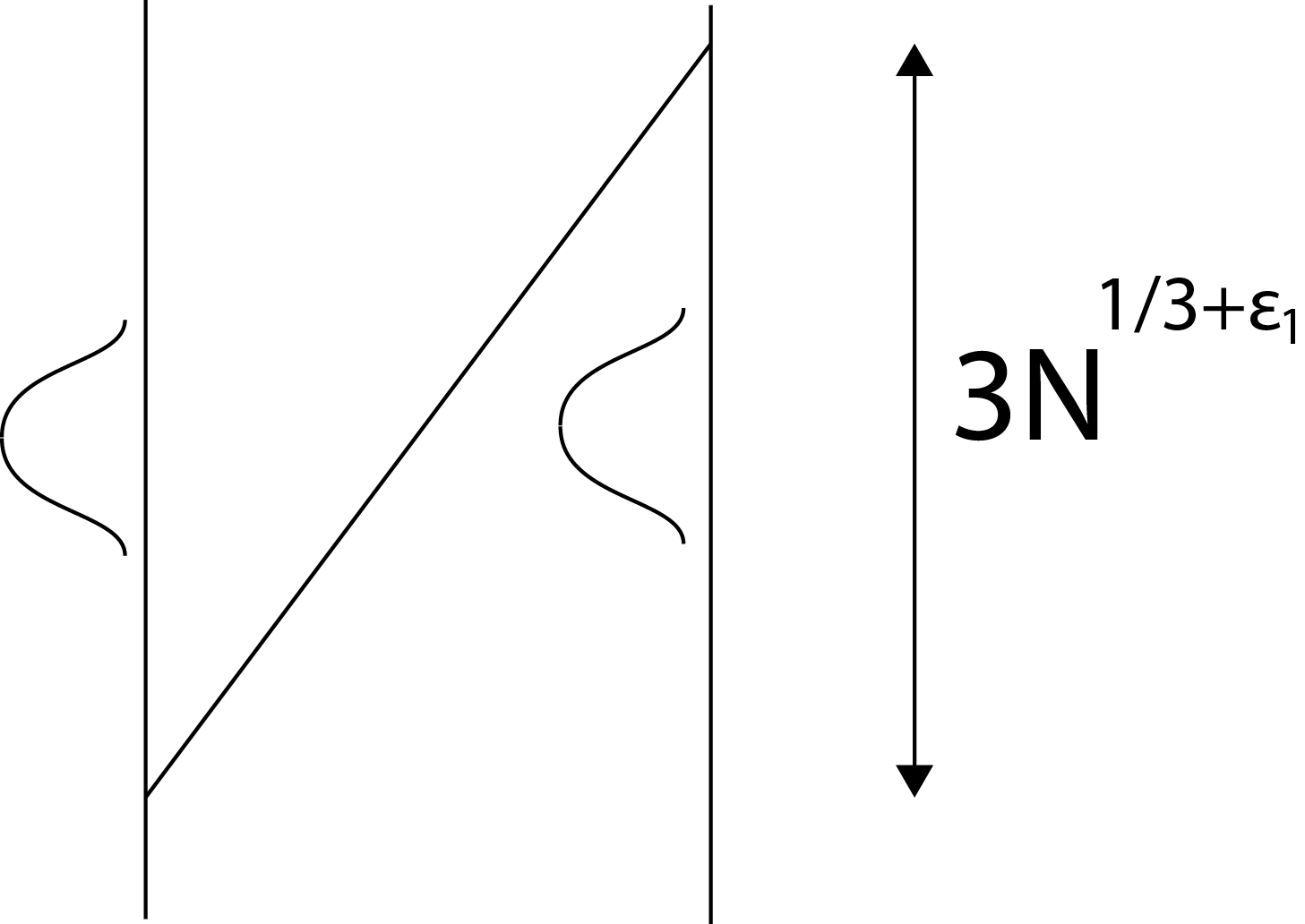}
  \end{center}
  \caption{The closest missing interaction}\label{dualrail}
\end{wrapfigure}

Assume there are two encoded qubits, and we are implementing a CPHASE gate on them.  Suppose the CPHASE gate is set of so that every point on the first 1 rail is connected to the nearest $3N^{1/3+\varepsilon_1}$ points on the second 1 rail.  Let the packet be localized at some point $x_0$ inside the CPHASE gate, suppose $x_0$ is at least a distance $3 N^{1/3+\varepsilon_1}$ for the gate entrance.  The ``closest'' missing 2-local interaction (the interaction where the farthest subsystem is closest to the packets) is the interaction between the two 1 rails that places the packets exactly in the middle (see figure).  In this case the missing interaction could have length $3N^{1/3+\varepsilon_1}+1$, so this interaction is at least a distance $\sim 1.5 N^{1/3+\varepsilon_1}$

Notice that as the packet traverses the gate, this missing interaction always stays at least a distance $\sim 1.5 N^{1/3+\varepsilon_1}$ from the center of the packets.  Each missing interaction can be thought of in this way.  Each missing interaction must be at least a distance $1.5 N^{1/3+\varepsilon_1}$ from the packet center at any given time.  As one subsystem for the 2-local interaction draws closer, the other one moves away.  In the following proof, $\widetilde{H}$ will be made up of interactions of this form.  


\begin{restatable}{thm}{Trotter1}\label{Trotter_1}
Suppose we have a superposition of packets $\ket{\psi}$ over $2m$ rails inside a particular gate block.  Suppose that for all $2 t'$ satisfying $0 \leq 2 t' \leq 2 t$, the translated packets are localized to at least a distance $d$ (see \cref{local_fig}).  Define $H_{{ \rm current}}$ to be the sum of the current gate Hamiltonian and the rail Hamiltonians.  Let $U_{2t}^p$ translate the packets a distance $2t$, as well as add an appropriate phase given the gate.  Then a sufficiently high degree positive polynomial $n(N)=N^{q_1}$:
\begin{equation}
\left|\left|\left(e^{-iH_{{ \rm current}}t}-U_{2t}^p \right)\ket{\psi} \right|\right|=O\left(poly(N) e^{-\frac{d^2}{2\Delta x^2}}+\frac{m t \Delta p^3}{N^3} +\frac{m^2 t^2}{n} \right)
\end{equation}
\end{restatable}
\begin{proof}

\begin{figure}
	\centering
	\includegraphics[width=5in]{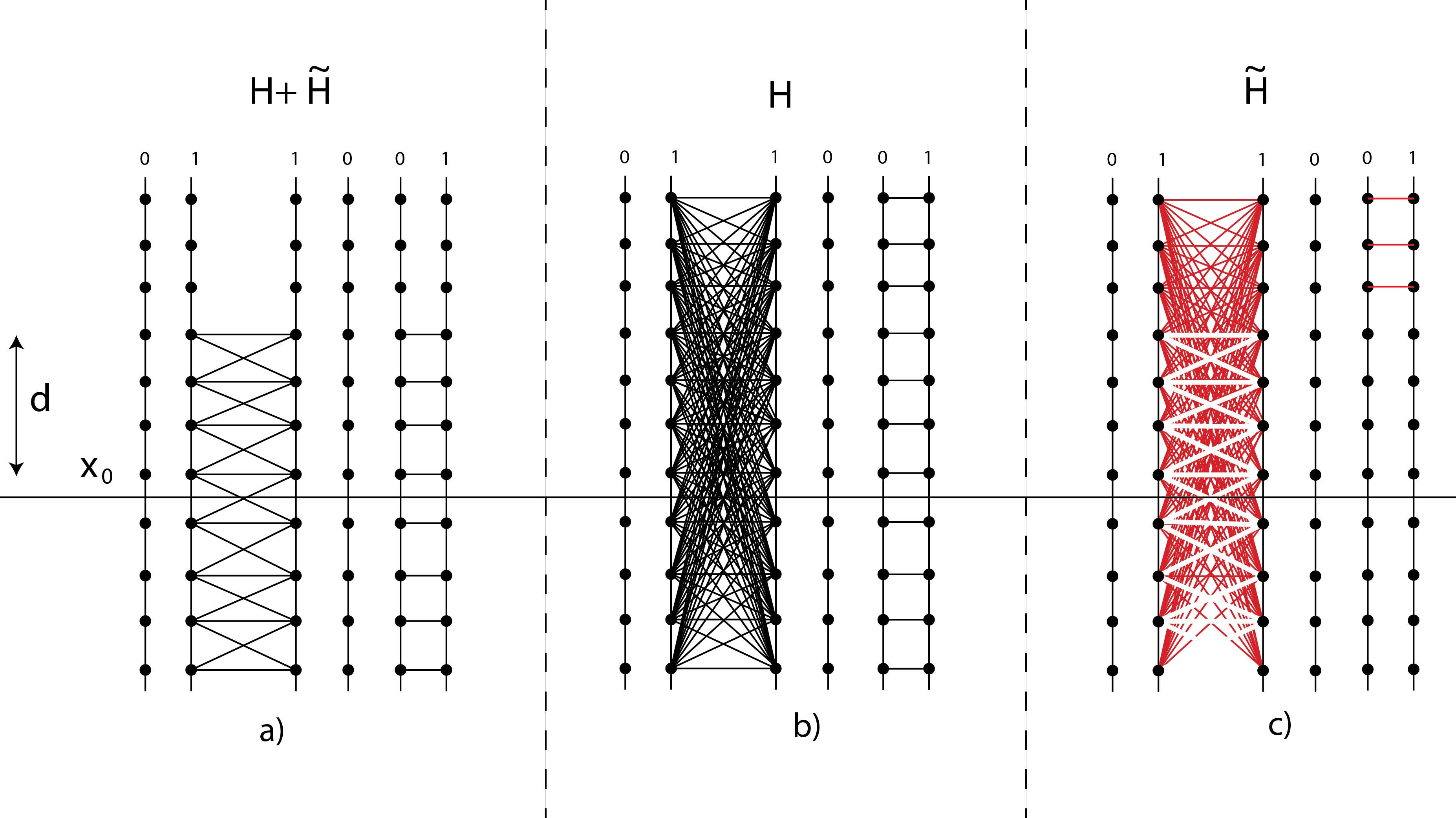}
	\caption[]{Pictorial representations of the relevant Hamiltonians for the proof.  In this circuit, we are implenting a CPHASE operation on the first two registers and an $\xx$ operation on the last one.  Note that the red lines are negative interaction terms.}\label{hamiltons}
\end{figure}

Let $H$ be the Hamiltonian obtained when the current gate Hamiltonians are extended to the whole chain.  Define $\widetilde{H}$ so that $H_{{ \rm current}}=H+\widetilde{H}$.  $\widetilde{H}$ contains negative correction terms that must be added to produce the actual dynamics (see \cref{hamiltons}).  Assuming $||H||_V \geq ||\widetilde{H}||_V$, if we Trotterize n times, it holds that:
\begin{equation}
\left|\left|\left(e^{-i(H+\widetilde{H})t}-\left(e^{-iHt/n}e^{-i \widetilde{H}t/n} \right)^n\right)\ket{\psi_0}\right|\right|=O \left(\frac{t^2 ||H||_V^2 e^{2||H||_V t/n}}{n} \right)
\end{equation}
$H$ consists of the ring Hamiltonian, and the extended gate Hamiltonians.  $H_{rings}$ has eigenvalues $2\cos(\frac{2 \pi p }{N})$, so $||H_{rings}^{2m}||_V \leq 2m$.  If $n$ is a large enough degree polynomial, $e^{||H|| t/n }\rightarrow 1$, so the above expression is:
\begin{equation}
O\left(\frac{m^2 t^2}{n} \right)
\end{equation}

Now we want to show that this Trotterization applies a phase and translates the packets, up to some approximation.  We want to bound:
\begin{equation}
\left|\left| \left(\left(e^{-iHt/n} e^{-i \widetilde{H}t/n} \right)^n-U_{2t}^p \right) \ket{\psi_0}\right|\right|
\end{equation}
Let $\ket{\psi_{2t'}}$ for $2t'\leq 2t$ be the same superposition, except translated and with added phase appropriate for the current gates.  We need bounds on 
\begin{equation}\label{eq_345}
\left|\left|\left(e^{-iHt/n}-U_{2t/n}^p\right)\ket{\psi_{2t'}}\right|\right|
\end{equation}
and on
\begin{equation}\label{eq_346}
\left|\left|\left( e^{-i\widetilde{H}t/n}-\mathbb{I}\right)\ket{\psi_{2t'}}  \right|\right|
\end{equation}

The distinction between the momentum basis and the position basis is important here.  \cref{eq_345} is bounded using the momentum basis, and \cref{eq_346} is bounded in the position basis.  Each time we switch from one to the other, we will incur some small error.  However, we will only need to switch $n=poly(N)$ times, so we will only incur an exponentially small error.

First recall that we have already shown that $e^{-iHt/n}$ is exactly $U_{2t/n}^p$ on relevant basis states, up to dispersion (by \cref{thm:disp}).  So, we can write:

\begin{equation}
\left|\left|\left(e^{-iHt/n}-U_{2t/n}^p\right)\ket{\psi_{2t'}}\right|\right|=O\left(\frac{m t \Delta p^3}{n N^3} \right)
\end{equation} 

Now we must treat the $\left|\left|(e^{-i \widetilde{H} t/n}-\mathbb{I})\ket{\psi_{2t'}}\right|\right|$ term.  Suppose that $\ket{\psi_{2t'}}$ is written in the RBD:
$$
\ket{\psi}=\sum_i c_i \ket{\psi_i}
$$
\noindent Once again we can verify that 
$$
\braket{\psi_i|\left(e^{-i\widetilde{H}t}-\mathbb{I} \right)^\dagger \left(e^{-i\widetilde{H}t}-\mathbb{I} \right)|\psi_j}=0
$$
when $i \neq j$, so we can bound the superposition with a bound on each relevant computational basis state.  

So now suppose that $\ket{\psi_{2t'}}$ is a computational basis state, and that it is at least a distance $d$ from the interactions in $\widetilde{H}$.  Suppose also that there are $|\widetilde{H}|=poly(N)$ many interactions in $\widetilde{H}$, all with only polynomial strength.  We can bound \cref{eq_346}:

\begin{equation}
\fl \left|\left|\left( e^{-i\widetilde{H} t/n}-\mathbb{I}\right)\ket{\psi_{2t'}} \right|\right|=\left|\left|\sum_{p=1}^\infty \frac{[-i\widetilde{H}t/n]^p}{p!}\ket{\psi_{2 t'}} \right|\right| =\left|\left|\sum_{p=1}^\infty \frac{[-i\widetilde{H}t/n]^p}{p!}\ket{\psi_{2t'}} \right|\right|_V
\end{equation}
\begin{equation}
\fl \leq ||\widetilde{H}\ket{\psi_{2t'}}||_V \sum_{p=1}^\infty \frac{(t/n)^p (||\widetilde{H}||_V)^{p-1}}{p!}=\frac{||\widetilde{H}\ket{\psi_{2t'}}||_V}{||\widetilde{H}||_V} \sum_{p=1}^\infty \frac{(t/n)^p (||\widetilde{H}||_V)^{p}}{p!}
\end{equation}
\begin{equation}
=\frac{||\widetilde{H} \ket{\psi}_{2t'}||_V}{||\widetilde{H}||_V} \left( e^{||\widetilde{H}||_V t/n}-1 \right)=O\left(poly(N)e^{|\widetilde{H}|\Phi t/n} e^{-\frac{d^2}{2\Delta x^2}} \right)
\end{equation}

\noindent by \cref{small_ham}.  Applying \cref{thm_kitaev} again, we get that:
\begin{equation}
\fl \left|\left|\left(\left(e^{-iHt/n} e^{-i\widetilde{H}t/n} \right)^n-U_{2t}^p\right)\ket{\psi_0}\right|\right|=O\left( poly(N)e^{poly(N)\Phi t/n}e^{-\frac{d^2}{2\Delta x^2}}+\frac{m t \Delta p^3}{N} \right)
\end{equation}
We will assume $n$ is a large enough polynomial in $N$ that $e^{poly(N)\Phi t/n}$ is small.  The added $e^{-\frac{d^2}{2\Delta x^2}}$ will always bring this term to zero in the limit of large $N$.  Adding in the trotterization term,
\begin{equation}
\left|\left|\left(e^{-iH_{{ \rm current}}t}-U_{2t}^p \right)\ket{\psi_0}\right|\right|=O\left(poly(N) e^{-\frac{d^2}{2\Delta x^2}}+\frac{m t \Delta p^3}{N^3}+\frac{m^2 t^2}{n}\right)
\end{equation}

\end{proof}

We will use this result to construct the following theorem, which treats the case we are interested in.
\Trotter

\begin{proof}
Note that the extra gate Hamiltonians commute with each other, since they act on different subsystems.  So, just like last time, assuming that $||H_{{ \rm current}}||_V \geq k ||H_j||_V$, it holds that:
\begin{eqnarray}
\fl \left|\left| \left( e^{-i(H_{{ \rm current}}+H_1+...+H_k)t}-\left( e^{-iH_{{ \rm current}}t/n'} e^{-i(H_1+H_2+...+H_k)t/n'} \right)^{n'}\right)\ket{\psi} \right|\right| \\ 
\fl =\left|\left| \left( e^{-i(H_{{ \rm current}}+H_1+...+H_k)t}-\left( e^{-iH_{{ \rm current}}t/n'} e^{-iH_1 t/n'} e^{-iH_2 t/n'}...e^{-iH_k t/n'} \right)^{n'}\right)\ket{\psi} \right|\right|\\
\fl \leq \left|\left| \left( e^{-i(H_{{ \rm current}}+H_1+...+H_k)t}-          \left( e^{-iH_{{ \rm current}}t/n'} e^{-iH_1 t/n'} e^{-iH_2 t/n'}...e^{-iH_k t/n'} \right)^{n'}\right) \right|\right|_V \label{eq_36}
\end{eqnarray}
\begin{equation}
=O\left(\frac{t^2||H_{{ \rm current}}||_V ^2 e^{2||H_{{ \rm current}}||_V t/n'}}{n'}\right)
\end{equation}
which can be seen by writing out the matrix exponentials and bounding terms of second order in $t/n'$

$\left|\left|H_{{ \rm current}}\right|\right|$ is at most $2m$, so this term is $O\left( \frac{t^2 m^2 e^{mt/n'}}{n'}\right)$.  We will eventually choose $mt/n' \rightarrow 0$ so this can be written as $O\left(\frac{t^2 m^2}{n'} \right)$.  

Now we want to show that these extra gate Hamiltonians can be ignored.  We want to bound the difference:
\begin{equation}
\left|\left|\left(\left( e^{-iH_{{ \rm current}}t/n'} e^{-iH_1 t/n'} ... e^{-iH_k t/n'} \right)^{n'}-U_{2t}^p \right)\ket{\psi_0} \right| \right|
\end{equation}
which can be written as:
\begin{equation}
\left|\left|\left(\left( e^{-iH_{{ \rm current}}t/n'} e^{-iH_1 t/n'} ... e^{-iH_k t/n'} \right)^{n'}-\left(U_{2t/n'}^p \mathbb{I} \mathbb{I}...\mathbb{I} \right)^{n'}\right)\ket{\psi_0} \right| \right|
\end{equation}
From \cref{Trotter_1}, we know that:
\begin{equation}
\fl \left|\left|\left(e^{-iH_{{ \rm current}}t/{n'}}-U_{2t/n'}^p\right)\ket{\psi_{2t'}}\right|\right|=O\left(poly(N) e^{-\frac{d^2}{2\Delta x^2}} +\frac{m t \Delta p^3}{n' N^3}+\frac{m^2 t^2 }{{n'}^2 n}\right)
\end{equation}

\noindent and just like in the proof of \cref{Trotter_1}:
\begin{equation}
\left|\left|\left( e^{-iH_j t/n'}-\mathbb{I}\right)\ket{\psi_{2t'}} \right|\right| \leq \frac{||H_j \ket{\psi}||_V}{||H_j||_V} \left( e^{||H_j||_V t/n'}-1 \right)
\end{equation}
Assuming $n$ and $n'$ are large enough
\begin{equation}
\left|\left|\left(e^{-i H_j t/n'}-\mathbb{I} \right)\ket{\psi_{2t'}} \right|\right| =O\left(poly(N) e^{-\frac{d^2}{2\Delta x^2}} \right)
\end{equation}
Applying \cref{thm_kitaev},
\begin{eqnarray}
\fl \left|\left|\left(\left(e^{-iHt/n'} e^{-iH_1t/n'} ... e^{-iH_k t/n'} \right)^{n'} - U_{2t}^p \right)\ket{\psi_0} \right|\right|=\nonumber\\
O\left(poly(N)e^{-\frac{d^2}{2\Delta x^2}}+\frac{m t \Delta p^3}{N^3}+\frac{m^2 t^2}{nn'} \right)
\end{eqnarray}
which then implies:
\begin{eqnarray}
\fl \left|\left|\left(e^{-i(H+H_1+...+H_k)t}-U_{2t}^p \right)\ket{\psi_0} \right|\right|=\nonumber\\
O\left(poly(N)e^{-\frac{d^2}{2\Delta x^2}}+\frac{m t \Delta p^3}{N^3}+\frac{m^2 t^2}{nn'} + \frac{t^2 m^2 }{n'}\right)
\end{eqnarray}

\end{proof}

\subsection{Big Picture}\label{subsec_6_5}

For clarity, we will include a note here on how exactly one would use these trotter theorems to derive our error bounds.  Suppose we have a quantum circuit with $g$ gate blocks on $m$ qubits.  Suppose our quantum circuit is designed to act on input $\ket{0}^{\otimes m}$, and let $U$ denote the circuit unitary.  Construct a spin network as described in the body of the paper to simulate this quantum circuit on a set of spin chains.  Let us say that each gate block has the same size on our spin rings (we can always make the interaction strength weaker if need be).  Denote the Hamiltonian for this spin network $H$, and let $T$ be the total time required for our packets to translate through the gate blocks, assuming they travel at the the group velocity.

Further, let $\ket{\widetilde{\psi_0}}$ be the quantum state corresponding to Gaussian packets initialized just outside the gate blocks centered on momentum $N/4$ heading toward the gate blocks.  Lastly, let $U_{{ \rm circuit}}$ be the encoded circuit unitary.  $U_{{ \rm circuit}}$ sends $\ket{\widetilde{\psi_0}}$ to the packet encoding of $U\ket{0}^{\otimes m}$, placed after all the gates.

Our end goal is to show that 
\begin{equation}\label{2355}
{ \rm Error}=\left|\left|\left(e^{-iH T}-U_{{ \rm circuit}}\right)\ket{\widetilde{\psi_0}}\right|\right|
\end{equation}
is small.  First, we must replace $\ket{\widetilde{\psi_0}}$ by the slightly unnormalized set of packets used in our proofs.  Let $\ket{\psi_0}$ be such an unnormalized state.  We have:
\begin{eqnarray}
\fl \left|\left|\left(e^{-iHT}-U_{{ \rm circuit}}\right)(\ket{\widetilde{\psi_0}}+\ket{\psi_0}-\ket{\psi_0})\right|\right| \leq \left|\left|\left(e^{-iHT}-U_{{ \rm circuit}}\right)\ket{\psi_0}\right|\right|+\left|\left|\ket{\widetilde{\psi_0}}-\ket{\psi_0}\right|\right|=  \nonumber \\
\left|\left|\left(e^{-iHT}-U_{{ \rm circuit}}\right)\ket{\psi_0}\right|\right|+O\left(\sqrt{\frac{m}{\Delta p}} \right)
\end{eqnarray}

Our strategy is to break up $e^{iHT}$ and $U_{{ \rm circuit}}$ into small pieces and apply \cref{thm_kitaev}.  Suppose that it takes $t_1$ for the packets to translate through the first block, $t_2$ for the packets to translate through the second block, etc.  Let us break up $U_{{ \rm circuit}}$ into $U_g U_{g-1}...U_1$ where each $U_i$ translates the packets from the beginning of gate block $i$ to the end of gate block $i$ as well as applies an appropriate phase.  More precisely, let us say that each $U_i$ translates the packets so that they are centered on the first interactions for the next block.  \cref{2355} can then be written as 
\begin{equation}
\fl \left|\left|\left(e^{-i H t_g}e^{-i H t_{g-1}}...e^{-i H t_1}-U_g U_{g-1}...U_1\right)\ket{\psi_0}\right|\right|  \leq \sum_{j=1}^g \left|\left|(e^{-iHt_j}-U_j)\prod_{k=1}^j U_k \ket{\psi_0}  \right|\right|
\end{equation}
We can break up $U_j$ further into two transient regions and the interior region, and break up $t_j$ into three different times accordingly.  
\begin{equation}
=\sum_{j=1}^g \left|\left| \left( e^{-i H t_j^{{ \rm trans}}}e^{-i H t_j^{{ \rm int}}}e^{-i H t_j^{{ \rm trans}}}-U_j^{{ \rm trans}}U_j^{{ \rm int}}U_j^{{ \rm trans}} \right)\prod_{k=1}^j U_k \ket{\psi_0} \right|\right|
\end{equation}

Now we can leverage \cref{thm_kitaev} to upper bound this quantity further:
\begin{eqnarray}
\leq \sum_{j=1}^g\left|\left|\left(e^{-iH t_j^{{ \rm trans}}}-U_j^{{ \rm trans}}\right) U_j^{{ \rm int}} U_j^{{ \rm trans}}\prod_{k=1}^j U_k\ket{\psi_0}\right|\right|+\nonumber\\ 
\fl \left|\left|\left(e^{-iH t_j^{{ \rm int}}}-U_j^{{ \rm int}}\right)  U_j^{{ \rm trans}}\prod_{k=1}^j U_k \ket{\psi_0}\right|\right|+
\left|\left|\left(e^{-iH t_j^{{ \rm trans}}}-U_j^{{ \rm trans}}\right) \prod_{k=1}^j U_k \ket{\psi_0}\right|\right| 
\end{eqnarray}

The first and third terms are similar, so they are taken care of in the same way using the transient bound.  Define $V_j^{{ \rm trans}}$ to be the translation operator without the added phase from the gate.  The third term, for instance, can be written as
\begin{eqnarray}
\left|\left|\left(e^{-i H t_j^{{ \rm trans}}}-U_j^{{ \rm trans}}+V_j^{{ \rm trans}}-V_j^{{ \rm trans}} \right)\prod_{k=1}^j U_k\ket{\psi_0}\right|\right| \leq \\
\left|\left|\left(e^{-i H t_j^{{ \rm trans}}}-V_j^{{ \rm trans}} \right)\prod_{k=1}^j \ket{\psi_0}\right|\right|+\left|\left|\left(V_j^{{ \rm trans}}-U_j^{{ \rm trans}}\right)\prod_{k=1}^j U_k\ket{\psi_0}\right|\right|
\end{eqnarray}
The first term is exactly our transient bound and the second term is the ``lost phase'' (see the proof of the transient bound), both of which have the same order.  

The second term is exactly the Trotter theorem.  We let $H=H_1+H_2$ and rewrite the second term in equation (128) as:
\begin{equation}
\fl \left|
\left|
\left[\left(
e^{-i H_1 t_j^{{ \rm int}}/n}e^{-i H_2 t_j^{{ \rm int}}/n}
\right)^n-\left(e^{-i H_1 t_j^{{ \rm int}}/n}e^{-i H_2 t_j^{{ \rm int}}/n}\right)^n+e^{-iH t_j^{{ \rm int}}}-U_j^{{ \rm int}}\right]\prod_{k=1}^j U_k \ket{\psi_0}\right|\right|
\end{equation}
\begin{eqnarray}
\fl \leq \left|\left|\left(
e^{-i H_1 t_j^{{ \rm int}}/n}e^{-i H_2 t_j^{{ \rm int}}/n}
\right)^n-e^{-iH t_j^{{ \rm int}}}\right|\right|+\nonumber \\
\left|\left|   \left[\left(
e^{-i H_1 t_j^{{ \rm int}}/n}e^{-i H_2 t_j^{{ \rm int}}/n}
\right)^n-U_j^{{ \rm int}}\right]\prod_{k=1}^j U_k \ket{\psi_0}\right|\right|
\end{eqnarray}
The first term is bounded using the standard Trotter type bound, the second term is bounded using \cref{thm_kitaev} and the fact that the ``bad'' interactions are far away from the packets at every point in time.

\end{document}